\documentclass[aps,pra,twocolumn,nofootinbib,superscriptaddress,10pt]{revtex4-2}

\usepackage{amssymb}        % includes amsfonts, so amsfonts is redundant
\usepackage{bm}             % bold for math italics/symbols
\usepackage{dsfont}         % alternative blackboard bold
\usepackage{mathrsfs}       % script font (\mathscr)
\usepackage{mathdots}       % ellipsis variants

\usepackage{amsthm}

\usepackage{graphicx}
\usepackage{csquotes}
\MakeOuterQuote{"}
\usepackage{mathtools,mleftright}
\mleftright
\usepackage{array}
\usepackage{tabularx}
\usepackage{booktabs}
\usepackage{cellspace}
\usepackage[colorlinks=true, allcolors=blue, breaklinks=true]{hyperref}
\hypersetup{pageanchor=false}

\newtheorem{theorem}{Theorem}
\newtheorem{lemma}{Lemma}

\newtheorem{proposition}{Proposition}

\DeclarePairedDelimiter{\ket}{\lvert}{\rangle}
\DeclarePairedDelimiter{\bra}{\langle}{\rvert}
\DeclarePairedDelimiterX{\braket}[2]{\langle}{\rangle}{#1\,\delimsize\vert\,#2}
\newcommand{\ketbra}[2]{\ket{#1}\mkern-2mu\bra{#2}}
\newcommand{\proj}[1]{\ketbra{#1}{#1}}

\DeclarePairedDelimiterX{\dket}[1]{\lvert}{\rangle}{%
  #1\delimsize\rangle\mkern-3.5mu}
\DeclarePairedDelimiterX{\dbra}[1]{\langle}{\rvert}{%
  \mkern-3.5mu\delimsize\langle #1}
\DeclarePairedDelimiterX{\dbraket}[2]{\langle}{\rangle}{%
  \mkern-3.5mu\delimsize\langle #1\,\delimsize\vert\,#2%
  \delimsize\rangle\mkern-3.5mu}
\DeclarePairedDelimiterX{\dketbra}[2]{\lvert}{\rvert}{%
  #1\delimsize\rangle\mkern-3.5mu\delimsize\rangle
  \mkern-1.5mu
  \delimsize\langle\mkern-3.5mu\delimsize\langle #2}
\newcommand{\dproj}[1]{\dketbra{#1}{#1}}

\DeclareMathOperator{\tr}{tr}

\DeclareMathOperator{\rk}{rank}

\DeclareMathOperator{\Tr}{Tr}
\newcommand{\Cl}{\mathrm{Cl}}

\newcommand{\Ha}{\mathrm{Haar}}

\newcommand{\Var}{\mathrm{Var}}

\newcommand{\urp}{\mathrm{URP}}

\newcommand{\mmod}{\!\mod} 

\newcommand{\equad}{\,\hphantom{=}\,}   % equals-width space
\newcommand{\lsp}{\mkern 2mu}           % small positive space

\newcommand{\bbC}{\mathbb{C}}
\newcommand{\bbE}{\mathbb{E}}

\newcommand{\bbZ}{\mathbb{Z}}

\newcommand{\bfI}{\mathbf{I}}

\newcommand{\bbone}{\mathds{1}}

\newcommand{\bmk}{\bm{k}}
\newcommand{\bml}{\bm{l}}

\newcommand{\caB}{\mathcal{B}}
\newcommand{\caD}{\mathcal{D}}
\newcommand{\caE}{\mathcal{E}}

\newcommand{\caH}{\mathcal{H}}

\newcommand{\caL}{\mathcal{L}}
\newcommand{\caM}{\mathcal{M}}

\newcommand{\caO}{\mathcal{O}}
\newcommand{\caP}{\mathcal{P}}

\newcommand{\caS}{\mathcal{S}}
\newcommand{\caT}{\mathcal{T}}
\newcommand{\caU}{\mathcal{U}}

\newcommand{\caW}{\mathcal{W}}

\newcommand{\rme}{\mathrm{e}}
\newcommand{\rmi}{\mathrm{i}}

\newcommand{\rmH}{\mathrm{H}}

\newcommand{\rmU}{\mathrm{U}}

\newcommand{\scrI}{\mathscr{I}}

\newcommand{\bQ}{\bar{Q}}
\newcommand{\bbfI}{\bar{\bfI}}

\newcommand{\bcaM}{\bar{\mathcal{M}}}

\newcommand{\bPhi}{\bar{\Phi}}

\newcommand{\tcaT}{\tilde{\mathcal{T}}}

\newcommand{\tpsi}{\tilde{\psi}}
\newcommand{\eref}[1]{Eq.~\textup{(\ref{#1})}}
\newcommand{\Eref}[1]{Equation~\textup{(\ref{#1})}}
\newcommand{\eqsref}[2]{Eqs.~(\ref{#1}) and (\ref{#2})}
\newcommand{\Eqsref}[2]{Equations~(\ref{#1}) and (\ref{#2})}

\newcommand{\fref}[1]{Fig.~\ref{#1}}
\newcommand{\Fref}[1]{Figure~\ref{#1}}

\newcommand{\lref}[1]{Lemma~\ref{#1}}

\newcommand{\lsref}[2]{Lemmas~\ref{#1} and \ref{#2}}

\newcommand{\pref}[1]{Proposition~\ref{#1}}
\newcommand{\Pref}[1]{Proposition~\ref{#1}}
\newcommand{\psref}[2]{Propositions~\ref{#1} and \ref{#2}}

\newcommand{\sref}[1]{Sec.~\ref{#1}}

\newcommand{\tref}[1]{Table~\ref{#1}}

\newcommand{\thref}[1]{Theorem~\ref{#1}}
\newcommand{\Thref}[1]{Theorem~\ref{#1}}
\newcommand{\thsref}[2]{Theorems~\ref{#1} and \ref{#2}}

\newcommand{\rcite}[1]{Ref.~\cite{#1}}
\newcommand{\rscite}[1]{Refs.~\cite{#1}}

\begin{document}
\title{Optimal Shadow Estimation with Minimal Measurement Settings}
\author{Zhiyao Yang}
\author{Datong Chen}
\affiliation{State Key Laboratory of Surface Physics, Department of Physics, and Center for Field Theory and Particle Physics, Fudan University, Shanghai 200433, China}

\affiliation{Institute for Nanoelectronic Devices and Quantum Computing, Fudan University, Shanghai 200433, China}

\affiliation{Shanghai Research Center for Quantum Sciences, Shanghai 201315, China}

\author{Huangjun Zhu}
\email{zhuhuangjun@fudan.edu.cn}	
\affiliation{State Key Laboratory of Surface Physics, Department of Physics, and Center for Field Theory and Particle Physics, Fudan University, Shanghai 200433, China}

\affiliation{Institute for Nanoelectronic Devices and Quantum Computing, Fudan University, Shanghai 200433, China}

\affiliation{Shanghai Research Center for Quantum Sciences, Shanghai 201315, China}

\affiliation{Hefei National Laboratory, Hefei 230088, China}

	\begin{abstract}
Shadow estimation is a powerful framework for predicting quantum properties from randomized measurements. While $3$-design protocols achieve optimal worst-case performance, the minimal number of measurement bases required for such optimality has remained open. Here we prove that $\Theta(d^2)$ measurement bases are both necessary and sufficient for worst-case optimal shadow estimation and construct an explicit basis family. In stark contrast, any state $2$-design already suffices for average-case optimality: the mean squared shadow norm of normalized observables is bounded by a universal constant, and we prove strong concentration for Haar-random states, yielding constant sample complexity for generic pure-state fidelity estimation. Easily implementable $2$-designs---from mutually unbiased bases, cyclic measurements, or shallow $\mathcal{O}(\log n)$-depth circuits---enable optimal average-case protocols with remarkably simple measurement strategies. Our results establish a fundamental complexity separation: worst-case estimation requires $\Theta(d^2)$ bases, whereas average-case performance requires only $\Theta(d)$ bases, with broad implications for quantum information theory and near-term experiments.		
	\end{abstract}
	
	\maketitle
	
	%=============================================
	% INTRODUCTION
	%=============================================
	
	%~\cite{Haah17,Cramer10}
	
\emph{Introduction}---Characterizing quantum systems efficiently is a central challenge in quantum science and technology. Full quantum state tomography demands resources that scale exponentially with system size, motivating the development of more targeted approaches. The classical shadow framework~\cite{HuangKP20} provides a powerful alternative: randomized measurements paired with classical post-processing can predict key properties of an unknown quantum state---including expectation values, fidelities, and entanglement witnesses---without full state reconstruction. Shadow estimation has since found broad applications in fidelity estimation~\cite{Elben20,YangCZ25,Eisert20,Kliesch21}, entanglement detection~\cite{Brydges19,ElbenAK20,Zhou20,Neven21,Liu23,YiLZ26}, and Hamiltonian learning~\cite{Bairey19,Anshu20,Hadfield22,Huang23}, with experimental demonstrations on photonic~\cite{Struchalin21, Zhang21}, superconducting~\cite{HuangB22,Hu25}, and trapped-ion~\cite{Brydges19,Stricker22} platforms. In practice, the number of distinct measurement bases directly determines the calibration and implementation cost of a shadow protocol, making the identification of minimal measurement schemes a pressing concern.

A shadow estimation protocol is specified by a measurement ensemble $\caE$ of weighted pure states forming a rank-one positive operator-valued measure (POVM)~\cite{HuangKP20,InnocentiLP23,Nguyen22}. Its efficiency is governed by the \emph{shadow norm} $\|O\|_\caE$ of the target observable~$O$, which characterizes the single-shot estimation variance: to estimate an expectation value within additive error~$\varepsilon$, $\caO(\|O\|_\caE^2/\varepsilon^2)$ samples suffice regardless of the unknown state. Protocols based on unitary 3-designs---such as the multiqubit Clifford group~\cite{Zhu17,Webb16}---achieve the optimal worst-case shadow norm $\|O\|_\caE^2\leq 3\|O\|_2^2$~\cite{HuangKP20}, but require large, highly structured ensembles, which are challenging to implement on near-term hardware. Simpler schemes based on state 2-designs---such as complete sets of mutually unbiased bases (MUBs)~\cite{DURT10,Ivanovic81,Wootters89} and cyclic measurements~\cite{Gonzalez25}---are far more practical, yet their worst-case squared shadow norm grows linearly in~$d$, rendering them suboptimal.

%Wang25,Wan23,

Despite rapid recent progress on shadow estimation based on various measurement ensembles~\cite{Elben23,HuangKP20,Hu23,Bu24,Hu22,Liu24,Zhang24,WangC24,Park25,Wu26,Ippoliti24,West26}, existing approaches either fail to achieve uniformly efficient estimation for general observables or rely on measurements with superpolynomially many outcomes. A fundamental question therefore remains open: \emph{What is the minimal measurement complexity required for optimal shadow estimation, and can simple measurements still be efficient in typical scenarios?}

Here we provide a complete answer, revealing a striking complexity separation. First, we prove that $\Theta(d^2)$ measurement bases are both necessary and sufficient for worst-case optimal shadow estimation, and construct an explicit basis family from phase 3-designs and MUBs. Second, we show that any state 2-design already achieves average-case
optimal performance with only $\Theta(d)$ bases: the mean squared shadow
norm over arbitrary target observables is bounded by $\left(2\sqrt{2}+1\right)\|O\|_2^2$
(\thref{thm:AvgShNormUB}). For normalized observables, this bound
evaluates to a constant independent of $d$, close to the 3-design optimum. The worst-case suboptimality of 2-designs thus vanishes in typical scenarios. Third, for fidelity estimation of Haar-random target states, we establish strong concentration of the shadow norm, yielding constant sample complexity for generic pure states even with simple 2-design measurements.

%Despite intensive recent study~\cite{HuangKP20,InnocentiLP23,Helsen23,chenZ24,Hu23,Bu24,Tran23,Li25,Nguyen22,Chen21,Elben23}

Our results reveal a fundamental $d^2$-versus-$d$ gap between worst-case and average-case measurement complexities, demonstrating that easily implementable ensembles---including complete sets of MUBs, cyclic measurements, and shallow $\caO(\log n)$-depth circuits \cite{Cleve16,Bertoni24,Hu25,Schuster25}---already realize optimal protocols for typical tasks. Meanwhile, our results highlight the foundational significance of MUBs and
symmetric informationally complete (SIC) measurements~\cite{Fuchs17,Zauner11,Renes04} in quantum learning.

	%=============================================
	% PRELIMINARIES
	%=============================================
	
\emph{Preliminaries}---Let $\caH$ be a $d$-dimensional complex Hilbert space ($d\geq 2$). Denote by $\caL(\caH)$ and $\caL^{\rmH}_0(\caH)$ the spaces of linear operators and traceless Hermitian operators on 
$\caH$, respectively;  denote by $\caD(\caH)$ the set of density operators, and by $\caP(\caH)\subset\caD(\caH)$ the subset of pure states. For a general operator $O\in \caL(\caH)$, we denote by $O_0=O-\tr(O)\bbone/d$ the traceless part of $O$, where $\bbone$ is the identity operator; for a pure state $\ket{\phi}\in\caH$, we write $\phi=\proj{\phi}$ and $\phi_0=\phi-\bbone/d$. The notation $\|\cdot\|_p$ is used for the Schatten $p$-norm and $\|\cdot\|$ for the operator norm.

Given a positive integer $t$, a state $t$-design is a weighted ensemble of pure states whose $t$-th moment reproduces the Haar average; it mimics Haar-random states up to $t$-th order statistics \cite{Zauner11,Renes04,Scott06,Ambainis07}. State 2-designs can be realized by SIC ensembles~\cite{Zauner11,Renes04,Fuchs17}, 
complete sets of MUBs~\cite{DURT10,Wootters89,Ivanovic81}, or cyclic constructions~\cite{Gonzalez25}; 
the set of multiqubit stabilizer states forms a state 
3-design~\cite{Kueng15,Zhu17,Webb16}. Unitary $t$-designs are defined analogously for ensembles of unitary operators.

A shadow estimation protocol is specified by a state ensemble $\caE=\{\ket{\phi_i},w_i\}_i$ with $w_i>0$ and ${\sum_i w_i \phi_i=\bbone/d}$, which ensures that $\{d\,w_i\phi_i\}_i$ forms a valid rank-one POVM on $\caH$. The associated measurement channel is
\begin{align}\label{eq:MeasCh}
	\caM_\caE(O) = d\sum_i w_i \tr(\phi_i O)\,\phi_i.
\end{align}
When $\caE$ is informationally complete (IC), i.e.,  $\{\phi_i\}_i$ spans $\caL(\caH)$, the channel $\caM_\caE$ is invertible. The inverse $\caM_\caE^{-1}$ acts as the reconstruction map: each snapshot $\hat{\rho}_i \coloneqq \caM_\caE^{-1}(\phi_i)$ satisfies $\bbE[\hat{\rho}_i] = \rho$, making it an unbiased estimator of $\rho$. Consequently, $\tr(O \hat{\rho}_i)$ is an unbiased estimator of $\tr(O\rho)$ for any observable $O$. 

The sample complexity of estimating $\tr(O\rho)$ to accuracy $\varepsilon$ with constant success probability scales as $\|O_0\|_\caE^2/\varepsilon^2$~\cite{HuangKP20}, where the \emph{squared shadow norm} is defined as follows. For a traceless observable $O_0\in\caL^{\rmH}_0(\caH)$ (the traceless part of $O$), the state-dependent variant is
\begin{align}\label{eq:normSD}
\|O_0\|_{\caE,\rho}^2 \coloneqq d\sum_i w_i \tr(\phi_i\rho)\bigl[\tr\bigl(O_0\,\hat{\rho}_i\bigr)\bigr]^2,
\end{align}
and the state-independent (worst-case) one is $\|O_0\|_\caE^2 \coloneqq \max_{\rho\in\caD(\caH)} \|O_0\|_{\caE,\rho}^2$. The worst-case norm provides a universal guarantee for all unknown states, while its $\rho$-dependent counterpart quantifies performance for a given target state and can be substantially smaller. Henceforth, observables are assumed traceless except in the context of fidelity estimation or otherwise stated.

A key quantity in our analysis is the \emph{normalized $t$-th frame potential}
\begin{equation}\label{eq:FPt}
	\bPhi_t(\caE) \coloneqq D_{[t]}\,\Phi_t(\caE),\;\;
	\Phi_t(\caE) \coloneqq \sum_{i,j} w_i w_j \bigl[\tr(\phi_i\phi_j)\bigr]^t,
\end{equation}
where $D_{[t]}=\binom{d+t-1}{t}$ is the dimension of the symmetric
subspace of $\caH^{\otimes t}$. It is well known that
$\bPhi_t(\caE)\geq 1$, with equality if and only if $\caE$ forms a
state $t$-design~\cite{Renes04,Scott06,Ambainis07}. In this work, the case $t=3$ plays
the central role: $\bPhi_3(\caE)$ quantifies how closely a 2-design
approximates a 3-design and governs the performance gap in shadow
estimation. For any state 2-design, its magnitude is controlled by:
\begin{proposition}\label{pro:FP3UB}
	If $\caE$ is a state 2-design in dimension~$d$, then
	\begin{align}\label{eq:FP3UB}
		\bPhi_3(\caE)
		\leq \frac{(d+2)(d^2+2d-1)}{6d^2}.
	\end{align}
	In particular, $\bPhi_3(\caE)=\caO(d)$ for any state 2-design.
\end{proposition}
\Pref{pro:FP3UB} and other results are proved in Supplemental Material (SM) \cite{supp}, which includes \rscite{Hoggar82,Zhu16}.

%Welch74,
	%=============================================
	% WORST-CASE OPTIMAL PROTOCOLS
	%=============================================
	%$_j,_k_j$
	
\emph{Minimal optimal protocols in the worst-case setting}---We first establish a fundamental lower bound on the worst-case shadow norm achievable with a given number of measurement outcomes.
\begin{proposition}\label{pro:ShNormLUB}
	Suppose the ensemble $\caE$ consists of $K$ pure states in $\caH$ and induces an IC-POVM. Then
	\begin{align}\label{eq:ShNormLUB}
		\max_{\psi\in\caP(\caH)} \|\psi_0\|_\caE^2 \geq\frac{(d^2-1)^2}{Kd}.
	\end{align}
\end{proposition}
This result encodes a fundamental trade-off: fewer measurement outcomes necessarily imply larger worst-case shadow norms. Achieving optimal worst-case performance---meaning $\|O\|_\caE^2\leq C\|O\|_2^2$ for a universal constant $C$---requires at least $\Omega(d^3)$ outcomes, or equivalently $\Omega(d^2)$ orthonormal bases, given that each basis corresponds to $d$ outcomes. Incidentally, when $\caE$ forms a state 2-design, the worst-case shadow norm satisfies $\|O\|_\caE^2\leq (d+1)\|O\|_2^2$~\cite{InnocentiLP23,Zhang24,WangC24}. For typical 2-designs consisting of $\caO(d^2)$ states---such as SICs and complete sets of MUBs---the lower and upper bounds match up to constant factors, demonstrating that standard 2-design measurements are suboptimal for worst-case performance. We now show that the $\Omega(d^2)$-basis lower bound is tight.
\begin{theorem}\label{thm:MMSoptWC}
	In dimension $d$, $\Theta(d^2)$ measurement bases are both necessary and sufficient for worst-case optimal shadow estimation.
\end{theorem}

Our explicit construction achieving this scaling rests on two ingredients: phase 3-designs and MUBs. A pure state $\ket{\phi}\in \caH$ is a phase state with respect to the computational basis $\{\ket{k}\}_k$ if all its entries have the same magnitude. A \emph{phase $t$-design}  is a finite ensemble of phase states whose $t$-th moment matches that of the uniform random-phase (URP) ensemble
\begin{equation}\label{eq:URPmain}
	\caT = \left\{\frac{1}{\sqrt{d}}\sum_{k=0}^{d-1}\rme^{\rmi\varphi_k}\ket{k}\colon \rme^{\rmi\varphi_k}\sim \rmU(1)\right\},
\end{equation}
which geometrically forms a $d$-dimensional torus. Intuitively, a phase $t$-design captures the same statistical correlations as fully random phases up to order~$t$, using only finitely many states. Mathematically, it is equivalent to a projective toric $t$-design studied in \rcite{Iosue24}. 
As shown in the End Matter, a phase 2-design can be constructed from $3p$ bases and a phase 3-design from $7p^2$ bases, where $p$ is any prime satisfying $p \geq \max\{d,5\}$ and can be chosen such that $p \leq \max\{2d,5\}$. When $d$ is not divisible by~$3$, the number of bases required for a phase 3-design can be reduced to~$3p^2$.
Since any phase $t$-design (with $t\geq 1$) is also a state 1-design, it gives rise to a valid POVM.

The construction proceeds as follows. Consider $N\geq 2$ MUBs $\caB_1,\caB_2,\ldots,\caB_N$. For each basis $\caB_j$, we construct a phase 3-design $\caE_{\caB_j}$ with respect to $\caB_j$. The combined ensemble $\caE_N=\bigsqcup_{j=1}^N \caE_{\caB_j}/N$ assigns uniform probability across all sub-ensembles and comprises at most $7Np^2$ bases in total. MUBs ensure that the combined ensemble probes complementary degrees of freedom: while a single phase 3-design can only access coherence relative to its reference basis, MUBs guarantee that no observable direction is left unexplored.
\begin{theorem}\label{thm:MMSoptWCMUB}
Suppose $\caE_N$ is the ensemble constructed from $N$ phase 3-design ensembles based on MUBs
with $2\leq N\leq d+1$, and $O\in\caL^{\rmH}_0(\caH)$. Then
	\begin{align}\label{eq:MMSoptWCMUB}
		\|O\|_{\caE_N}^2 \leq \frac{3N^2+N}{(N-1)^2}\|O\|_2^2 \leq 14\|O\|_2^2.
	\end{align}
\end{theorem}
Already for $N=2$ MUBs, we have $\|O\|_{\caE_2}^2 \leq 14\|O\|_2^2$, achieving constant worst-case shadow norm with $\Theta(d^2)$ bases---matching the lower bound from \pref{pro:ShNormLUB} up to a constant factor. The bound improves monotonically with $N$: for $N=d+1$ (a complete set of MUBs), the prefactor reads $3+7/d+4/d^2$, which is quite close to the optimal value of~$3$ achieved by state 3-designs.
In practice, $N=2$ or $3$ MUBs strike a favorable balance between measurement overhead and shadow norm, as adding MUBs rapidly decreases the prefactor at first but yields diminishing returns thereafter. The ensemble $\caE_N$ is not a state 2-design unless $N=d+1$; nevertheless, the measurement channel $\caM_{\caE_N}$ and reconstruction map $\caM_{\caE_N}^{-1}$ both admit simple closed forms (see the End Matter), enabling straightforward construction of unbiased estimators.

\begin{figure}[t]
\centering
\includegraphics{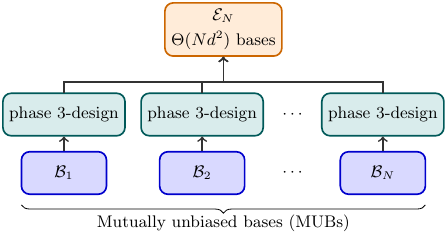}  
\vspace{-5pt} 
\caption{Construction of the optimal measurement ensemble $\caE_N$ for worst-case shadow estimation from mutually unbiased bases (MUBs) and phase 3-designs. Each phase 3-design is constructed from $\Theta(d^2)$ bases. The ensemble $\caE_N$ is composed of $\Theta(d^2)$ bases whenever $N=\caO(1)$ but can achieve optimal scaling in the worst-case shadow norm as long as $N\geq 2$. }  
\label{fig:concept}
\end{figure}

	%=============================================
	% AVERAGE-CASE PERFORMANCE
	%=============================================

    \emph{Minimal optimal protocols for average performance}---We now turn to average performance and show that far fewer measurement bases suffice. We call a protocol average-case optimal if the mean squared shadow norm over a unitary orbit---observables sharing the spectrum of $O$ but with arbitrary eigenbases---is bounded by $c\|O\|_2^2$ for a dimension-independent constant $c$.

For $O\in\caL(\caH)$ and a unitary ensemble $\caU$ on $\rmU(\caH)$, define the orbit
\begin{align}
	\Xi(O,\caU)\coloneqq \left\{UOU^\dagger \mid U\sim\caU\right\}.
\end{align}
Let $\|\Xi(O,\caU)\|_\caE^2$ denote the mean squared shadow norm averaged over $\caU$, and $\|\Xi(O,\caU)\|_{\caE,\rho}^2$ its state-dependent counterpart. Note that $\caU$ defines only the observable class over which performance is averaged; it plays no role in the measurement protocol. For a unitary 2-design $\caU$ and a state 2-design $\caE$~\cite{InnocentiLP23,chenZ24}, we have
\begin{align}\label{eq:SDopt}
	\|\Xi(O,\caU)\|_{\caE,\rho}^2 = \frac{d+1}{d}\|O\|_2^2,
\end{align}
so any state 2-design is already optimal in the state-dependent setting. However, this is of limited practical interest because $\rho$ is usually unknown. The state-independent norm $\|\Xi(O,\caU)\|_\caE^2$, providing a universal guarantee regardless of $\rho$, can be much larger; the following theorem bounds this gap.
\begin{theorem}\label{thm:AvgShNormUB}
	Suppose $\caE$ is a state $2$-design, $O\in\caL^{\rmH}_0(\caH)$, and $\caU$ is a unitary $4$-design. Then
	\begin{align}
		\|\Xi(O,\caU)\|_\caE^2 &\le \left(1 + \sqrt{\frac{24}{d}\big[\bPhi_3(\caE){-}1\big] + 4}\,\right)\|O\|_2^2 \nonumber\\
		&\le \left(2\sqrt{2}+1\right)\|O\|_2^2. \label{eq:AvgShNormUB}
	\end{align}
\end{theorem}
The 4-design assumption on $\caU$ (satisfied, e.g., by the Haar measure) ensures sufficient moment control over the orbit.  The parameter $\bPhi_3$ quantifies deviation from a 3-design [$\bPhi_3 = \caO(d)$ for any state 2-design by \pref{pro:FP3UB}]: at $\bPhi_3=1$ the bound recovers the 3-design optimum $3\|O\|_2^2$, while even the worst 2-design incurs at most a constant-factor overhead.

State 2-designs can be constructed from $d^2+\caO\left(d^{1.525}\right)$ states~\cite{Jasper25} or $\caO(d)$ orthonormal bases~\cite{Roy07,Gross07}. Explicitly, mixing a fixed basis~$\caB$ [with probability $1/(d+1)$] with a phase 2-design over~$\caB$ [with probability $d/(d+1)$] yields a valid state 2-design. Meanwhile, a phase 2-design requires at most $6d$ bases (\pref{pro:Phase2design} in the End Matter). Conversely, at least $d^2$ states or $d+1$ bases are necessary to construct an IC ensemble and guarantee finite shadow norms for all observables in $\Xi(O,\caU)$, assuming that the orbit spans $\caL^{\rmH}_0(\caH)$.

\begin{theorem}\label{thm:MMSoptAC}
	In dimension $d$, $\Theta(d)$ measurement bases---equivalently $\Theta(d^2)$ measurement outcomes---are necessary and sufficient for average-case optimal shadow estimation.
\end{theorem}
The contrast is striking: average-case optimal shadow estimation requires only $\Theta(d)$ bases, versus $\Theta(d^2)$ in the worst case---a fundamental quadratic gap. Notably, ensembles as simple as complete sets of MUBs \cite{DURT10,Wootters89,Ivanovic81}, cyclic measurements~\cite{Gonzalez25}, or shallow $\caO(\log n)$-depth circuits~\cite{Cleve16} already form state 2-designs and thus achieve average-case optimal performance, enabling practical protocols with unexpectedly economical measurement strategies. Moreover, $d+1$ bases suffice to realize average-case optimal shadow estimation whenever $d$ is a prime power, since a complete set of MUBs exists in every prime-power dimension \cite{DURT10,Wootters89,Ivanovic81}. Furthermore, if a SIC exists in every dimension---as conjectured and supported by strong evidence \cite{Fuchs17,Zauner11,Renes04,Scott10}---then $d^2$ outcomes suffice to realize a minimal optimal protocol.

	%=============================================
	% FIDELITY ESTIMATION

\emph{Fidelity estimation}---Predicting the fidelity $\tr(\rho\psi)$
between an unknown state $\rho$ and a known target state $\psi$ is
among the most important applications of shadow estimation, central to
state verification, benchmarking, and
certification~\cite{Eisert20,Kliesch21}. Here the observable is the
projector~$\psi$, so sample complexity is governed by
$\|\psi_0\|_\caE^2$ with $\psi_0=\psi-\bbone/d$.
\Thref{thm:AvgShNormUB} immediately bounds the mean squared shadow
norm for Haar-random target states. The following proposition further
establishes a large-deviation guarantee, whose proof requires an
independent and considerably more involved argument (see SM \sref{sup:avg_norm_up}).
\begin{proposition}\label{pro:ShNormFidLD}
	Suppose $\caE$ is a state $2$-design, $\psi$ is a Haar-random pure
	state, and $k>0$. Then
	\begin{gather}
		\underset{\psi\sim\Ha}{\bbE}\left\|\psi_0\right\|_\caE^2
		\leq\left(2\sqrt{2}+1\right), \label{eq:ShNormFidUB} \\
		\Pr\left\{\|\psi_0\|_\caE^2
		\ge\sqrt{48[1+k\xi(\caE)]}-1\right\}
		\le\frac{1}{(d+6)^3k^2}, \label{eq:ShNormFidLD}
	\end{gather}
	where
	\begin{align}
		\xi(\caE)\coloneqq\sqrt{102[\bPhi_3(\caE)-1]+6}<\sqrt{17d}. \label{eq:Xi}
	\end{align}
\end{proposition}
For SIC ensembles and complete sets of MUBs, the normalized third frame potentials read
\begin{align}\label{eq:bPhi3SICMUB}
	\bPhi_3(\caE_\mathrm{SIC})=\frac{(d+2)(d+3)}{6(d+1)},\;\,
	\bPhi_3(\caE_\mathrm{MUB})=\frac{(d+1)(d+2)}{6d},
\end{align}
each scaling as $d/6$ in the large-$d$ limit and attaining near-maximal values among 2-designs (cf.\ \pref{pro:FP3UB}). Nevertheless, \fref{fig:AShNormFid} demonstrates that these canonical 2-designs already closely match the 3-design benchmark for average-case performance.

\begin{figure}[tb]
	\centering
	\includegraphics[width=0.45\textwidth]{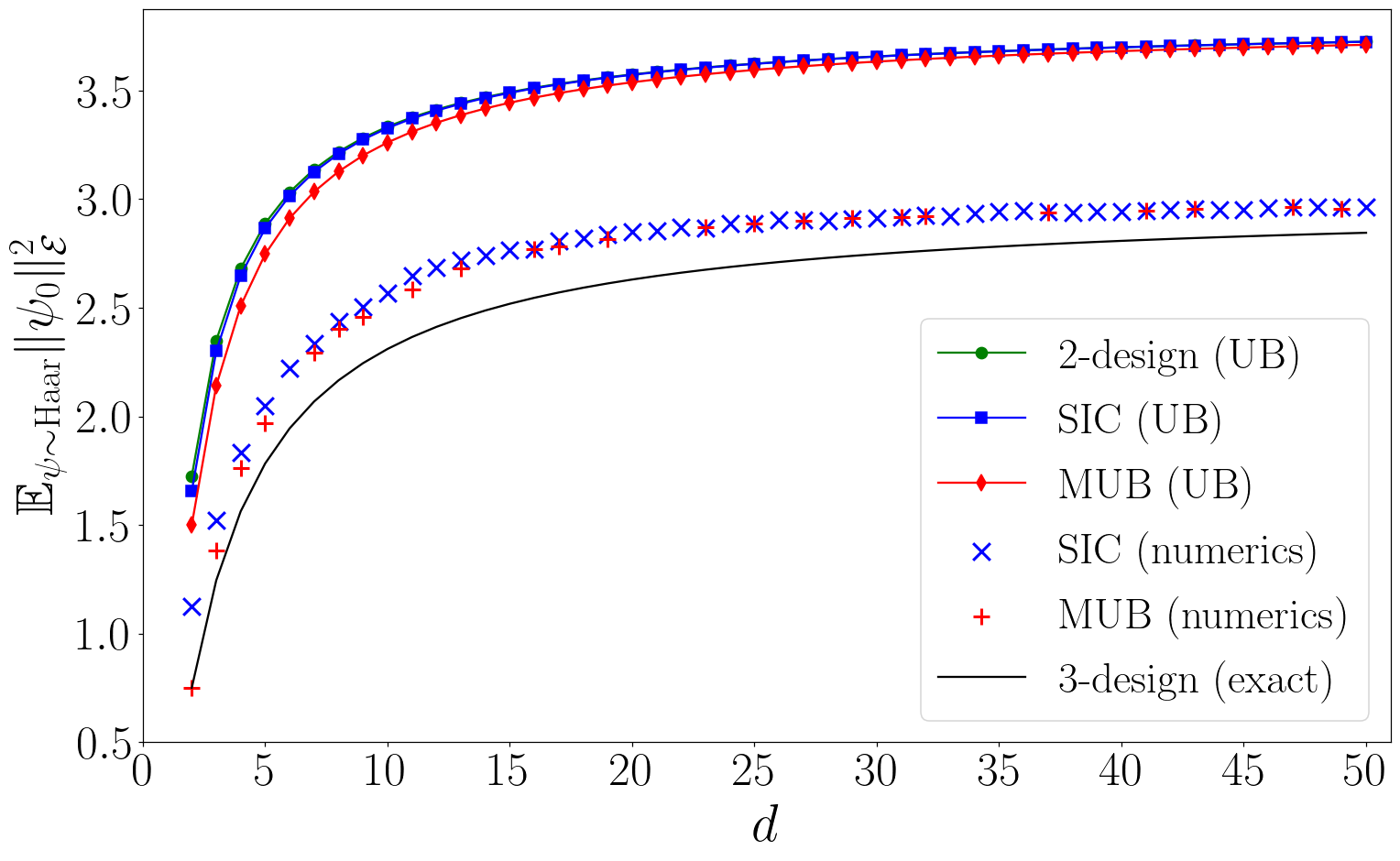}
    \vspace{-6pt} 
	\caption{Mean squared shadow norm for fidelity estimation of Haar-random pure states versus dimension $d$. SIC and MUB results are compared with upper bounds from \eref{eq:ShNormFidUB}, the exact 3-design result, and the worst-case 2-design bound [based on \pref{pro:FP3UB} and \eref{eq:ShNormFidUB}]. Notably, the average performance across all 2-designs, including SICs and MUBs, nearly matches the 3-design benchmark.  }
	\label{fig:AShNormFid}
\end{figure}

Since $\bPhi_3(\caE)=\caO(d)$ for any 2-design (\pref{pro:FP3UB}), choosing $k=\caO\left(1/\sqrt{d}\lsp\right)$ in \eref{eq:ShNormFidLD} yields $\|\psi_0\|_\caE^2=\caO(1)$ except with probability $\caO(1/d^2)$. Hence, the shadow norm concentrates tightly around its mean for Haar-random~$\psi$, with a universally fast rate for all state 2-designs; dependence on $\bPhi_3$ only marginally slows convergence for larger~$\bPhi_3$. States with anomalously large shadow norms thus form a vanishingly small subset of the state space: generic pure-state fidelity estimation achieves constant sample complexity with any 2-design measurement, constructible from only $\caO(d)$ bases. This stands in stark contrast to the $\Theta(d^2)$ bases required for worst-case optimality.

For stabilizer measurements, by means of involved analysis, a universal upper bound on the shadow norm was established in \rscite{Mao25,ZhuMC24}, scaling with local dimension but not qudit number. However, that bound cannot imply constant sample complexity for generic fidelity estimation. Our concentration result is stronger in this regard: the shadow norm remains $\caO(1)$ with high probability for any state 2-design, independent of both the system dimension and local dimension. This prediction
is corroborated by numerical calculation on $n$-qudit stabilizer measurements, as illustrated 
in \fref{fig:AShNormLDstab} in the End Matter.

Beyond Haar-random states, the mean squared shadow norm also admits constant upper bounds for other physically motivated ensembles of target states: URP ensembles, relevant to Hamiltonian simulation, and Clifford orbits, relevant to randomized benchmarking and quantum error correction (\pref{pro:AShNormFidUB} in the End Matter).

	%=============================================
	% SUMMARY AND OUTLOOK
	%=============================================
\emph{Summary and outlook}---We have determined the fundamental measurement complexity of shadow estimation. In the worst case, $\Theta(d^2)$ bases are both necessary and sufficient for optimal performance; our explicit construction from phase 3-designs and MUBs achieves this without requiring exact state 3-design structure. In the average case, only $\Theta(d)$ bases suffice: any state 2-design yields a constant mean squared shadow norm governed by $\bPhi_3$. For fidelity estimation, the shadow norm concentrates tightly for Haar-random targets, implying constant sample complexity for generic instances. Together, these results reveal a sharp quadratic gap between worst-case and average-case complexities: $d^2$ versus $d$ bases. For a $10$-qubit system, this reduces the required number of bases from ${\sim}10^6$ to ${\sim}10^3$, demonstrating that the practical cost can be far lower than the worst-case bound suggests.

Looking forward, extending these results to multipartite settings and structured observable classes---such as local Hamiltonians, where problem-specific structure may yield further reductions---remains an important open direction. On the practical side, easily implementable 2-designs---realized via complete sets of MUBs \cite{Zhang24,WangC24}, cyclic measurements~\cite{Gonzalez25}, or shallow $\caO(\log n)$-depth circuits~\cite{Cleve16,Bertoni24,Hu25,Schuster25}---already achieve optimal average-case performance. This is especially relevant for decoherence-limited devices, where measurement overhead is a primary bottleneck. Finally, connections with Hamiltonian-driven shadow protocols~\cite{Hu22,Liu24} offer a promising route: the system's intrinsic dynamics could naturally generate structured measurements, enabling optimal shadow estimation without external randomization.

	%=============================================
	% ACKNOWLEDGMENTS
	%=============================================
    \emph{Acknowledgments}---We thank Changhao Yi, Xiaodi Li, and Xinyang Shu for inspiring discussions. This work is supported by Shanghai Science and Technology Innovation Action Plan (Grant No.\ 24LZ1400200), National Natural Science Foundation of China (Grant No.\ 92576101), 
    Quantum Science and Technology-National Science and Technology Major Project (Grant No.\ 2024ZD0300101), National Key Research and Development Program of China (Grant No.\ 2022YFA1404204), and Shanghai Municipal Science and Technology Major Project (Grant No.\ 2019SHZDZX01).

\let\oldaddcontentsline\addcontentsline
\renewcommand{\addcontentsline}[3]{}
\bibliography{ref}
\let\addcontentsline\oldaddcontentsline
	
\clearpage
\newpage
	%=============================================
	% END MATTER
	%=============================================
	\onecolumngrid
	\vspace{0.8cm}
	\begin{center}
		\rule{0.5\textwidth}{0.4pt}\\[6pt]
		{\large\textbf{End Matter}}\\[6pt]
		\rule{0.5\textwidth}{0.4pt}
	\end{center}
	\twocolumngrid

%=============================================
% APPENDIX A
%=============================================

\emph{Appendix A: Reconstruction map for the combined phase-design ensemble $\caE_N$}---Consider $N\geq 2$ MUBs $\caB_1,\caB_2,\ldots,\caB_N$ and the ensemble $\caE_N$ used in \thref{thm:MMSoptWCMUB}. Any operator $O\in\caL(\caH)$ can be decomposed as
\begin{align}
	O = \frac{\tr(O)\bbone}{d} + O_\perp + \sum_{j=1}^{N} O_{\caB_j},
\end{align}
where $O_{\caB_j}$ denotes the traceless component diagonal in basis~$\caB_j$, and $O_\perp$ is the component orthogonal to all $\caB_j$-diagonal operators, which vanishes when $N=d+1$. With respect to this decomposition, the reconstruction map takes the simple form
\begin{align}\label{eq:ReconMap}
	\caM_{\caE_N}^{-1}(O) = \frac{\tr(O)\bbone}{d} + dO_\perp + \sum_{j=1}^{N}\frac{Nd}{N-1}O_{\caB_j}.
\end{align}
For a complete set of MUBs ($N=d+1$), the ensemble $\caE_N$ forms a 2-design, and the reconstruction map reduces to the familiar form $\caM_{\caE_{d+1}}^{-1}(O) = (d+1)O - \tr(O)\bbone$ \cite{HuangKP20}.

\emph{Appendix B: Construction of nearly tight phase 2- and 3-designs from orthonormal bases}---To construct a phase 2-design (3-design), at least $\Theta(d^2)$ ($\Theta(d^3)$) states are required. Here we provide nearly tight constructions from orthonormal bases. In the case of phase 2-designs, our construction is a reformulation of the construction in \rcite{Gross07}, but in a much simpler language. This reformulation is crucial to generalization to phase 3-designs.

Given any positive integer $t$, define the function 
\begin{equation}
	f_t(x)\coloneqq \left\lfloor \frac{tx}{d}\right\rfloor,\quad x\in \scrI:= \{0,1,\ldots,d-1\}.
\end{equation}
Let $p$ be any prime that satisfies $p\geq\max\{d,3\}$. Then a phase 2-design can be constructed as follows:
\begin{equation}\label{eq:caT2}
	\begin{gathered}
		\caT_2\coloneqq \left\{\ket{\psi_2(a,b,c)} \colon a\in \bbZ_d,\, b\in \bbZ_3,\, c\in \bbZ_p\right\},\\
		\ket{\psi_2(a,b,c)}\coloneqq \frac{1}{\sqrt{d}} \sum_{x\in \scrI} \omega_d^{ax}\omega_3^{bf_2(x)}\omega_p^{cx^2}\ket{x},
	\end{gathered}
\end{equation}
where $\omega_k=\rme^{2\pi\rmi/k}$ for a positive integer $k$ is a primitive $k$-th root of unity;
note that $\caT_2$ is composed of $3p$ bases. When $d$ is odd, we have an alternative construction using only $2p$ bases:
\begin{equation}\label{eq:tcaT2}
	\begin{gathered}
		\tcaT_2\coloneqq\left\{\ket{\tpsi_2(a,b,c)} \colon a\in \bbZ_d,\, b\in \bbZ_2,\, c\in \bbZ_p\right\},\\
		\ket{\tpsi_2(a,b,c)}\coloneqq\frac{1}{\sqrt{d}} \sum_{x\in \scrI} \omega_d^{ax}\omega_2^{bx}\omega_p^{cx^2}\ket{x}.
	\end{gathered}
\end{equation}
\begin{proposition}\label{pro:Phase2design}
	The ensembles $\caT_2$ and $\tcaT_2$ defined in \eqsref{eq:caT2}{eq:tcaT2} form phase 2-designs.
\end{proposition}

Next, we turn to phase 3-designs. Let $p$ be any prime that satisfies $p\geq\max\{d,5\}$. Then a phase 3-design can be constructed as follows:
\begin{equation}\label{eq:caT3}
	\begin{gathered}
		\caT_3\coloneqq\left\{\ket{\psi_3(a,b,c)} \colon a\in \bbZ_d,\, b\in \bbZ_7,\, c_2, c_3\in \bbZ_p\right\},\\
		\ket{\psi_3(a,b,c)}\coloneqq\frac{1}{\sqrt{d}} \sum_{x\in \scrI} \omega_d^{ax}\omega_7^{bf_3(x)}\omega_p^{c_2x^2+c_3x^3}\ket{x};
	\end{gathered}
\end{equation}
note that $\caT_3$ is composed of $7p^2$ bases. When $d$ is not divisible by~$3$, we have an alternative construction using only $3p^2$ bases:
\begin{equation}\label{eq:tcaT3}
	\begin{gathered}
		\tcaT_3\coloneqq\left\{\ket{\tpsi_3(a,b,c)} \colon a\in \bbZ_d,\, b\in \bbZ_3,\, c_2, c_3\in \bbZ_p\right\},\\
		\ket{\tpsi_3(a,b,c)}\coloneqq\frac{1}{\sqrt{d}} \sum_{x\in \scrI} \omega_d^{ax}\omega_3^{bx}\omega_p^{c_2x^2+c_3x^3}\ket{x}.
	\end{gathered}
\end{equation}
\begin{proposition}\label{pro:Phase3design}
	The ensembles $\caT_3$ and $\tcaT_3$ defined in \eqsref{eq:caT3}{eq:tcaT3} form phase 3-designs.
\end{proposition}

\Pref{pro:Phase3design} is a simple corollary of \lref{lem:DesignPermutation} in SM \sref{sup:tDesignPhaseDesign}; \pref{pro:Phase2design} follows from similar but simpler reasoning. It is well known that, for any positive integer~$d$, there exists a prime $p$ with $p\leq 2d$~\cite{Ha08}. Therefore, for any dimension~$d$, a phase 2-design can be constructed from no more than $6d$ bases, while a phase 3-design can be constructed from no more than $28d^2$ bases. In the large-$d$ limit, it is possible to tighten these bounds.

\emph{Appendix C: Numerical evidence for concentration in stabilizer shadow estimation}---\Fref{fig:AShNormLDstab} plots the probability upper bound from \pref{pro:ShNormFidLD} for the event $\|\psi_0\|_{\caE}^2 \geq 9$. Here, $\psi$ denotes a Haar-random pure state of $n$ qudits with local dimension~$p$, and fidelity estimation is implemented via stabilizer measurements. These measurements form a Clifford orbit which, for odd prime~$p$, is a state 2-design but not a 3-design~\cite{Zhu17,Webb16,Kueng15}. The worst-case shadow norm satisfies $\|\psi_0\|_\caE^2\leq 2p-1$~\cite{Mao25}, with the bound approximately saturated by stabilizer states when $n\geq 2$. This worst-case scaling with~$p$ might suggest suboptimality for fidelity estimation. However, \fref{fig:AShNormLDstab} demonstrates that the large-deviation probability vanishes rapidly with increasing~$n$ for various choices of~$p$. The normalized third frame potential $\bPhi_3(\caE)=(p+1)(d+2)/[3(d+p)]$~\cite{ZhuMC24} approaches $(p+1)/3$ for large~$n$; together with \pref{pro:ShNormFidLD}, this yields an $\caO\left(1/d^2\right)$ tail bound that explains the observed concentration and confirms constant sample complexity for typical fidelity estimation. Related numerical observations were reported in \rcite{Mao25}; our \thref{thm:AvgShNormUB} and \pref{pro:ShNormFidLD} provide the rigorous theoretical underpinning that has been missing in the prior work.

\begin{figure}[tb]
	\centering
	\includegraphics[width=0.45\textwidth]{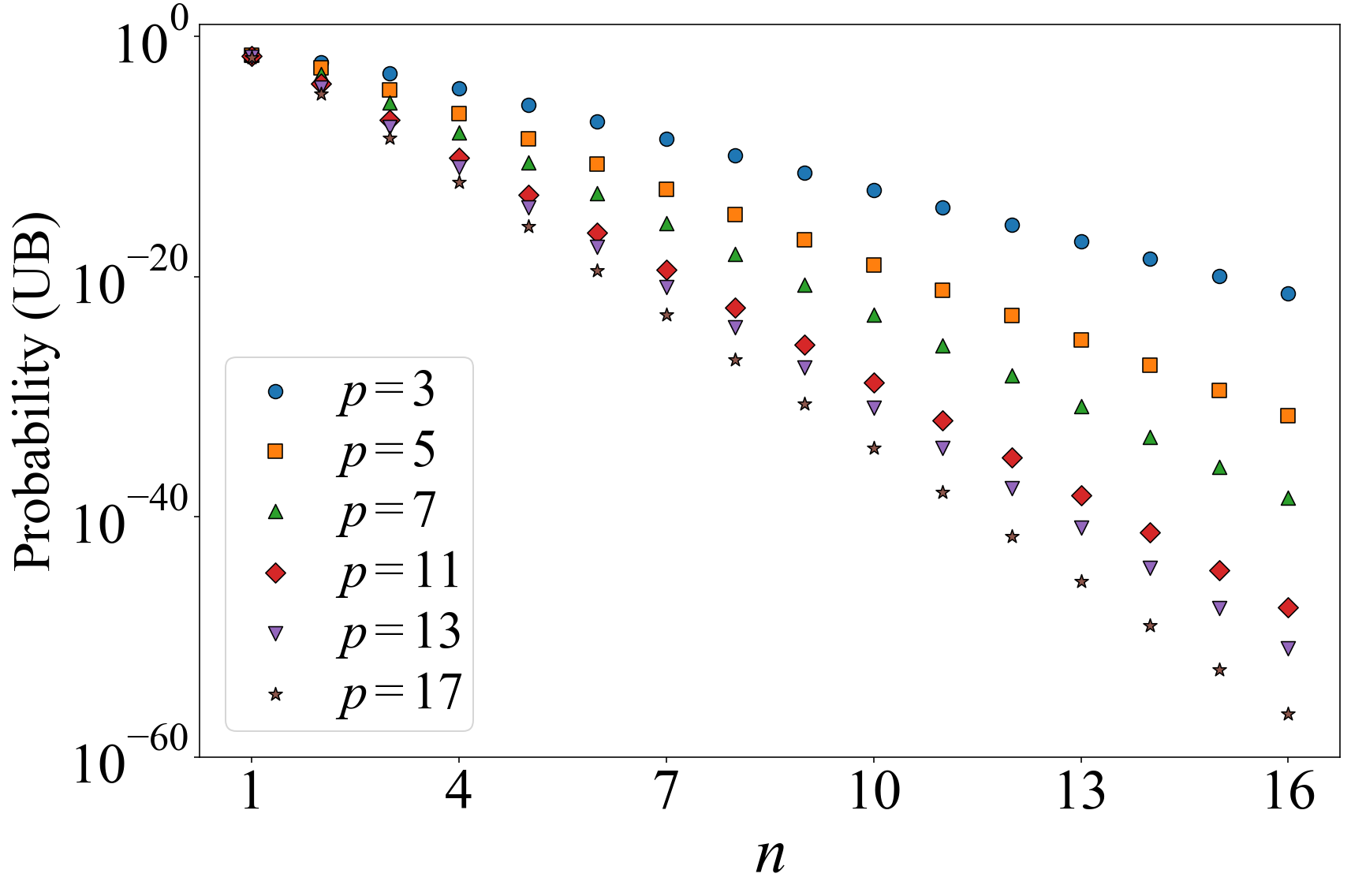}
	\caption{Probability that $\|\psi_0\|_\caE^2\geq 9$ for a Haar-random pure state~$\psi$ in fidelity estimation based on $n$-qudit stabilizer measurements with local dimension~$p$. The rapid decay confirms constant sample complexity for generic target states, in agreement with \thref{thm:AvgShNormUB} and \pref{pro:ShNormFidLD}. }
	\label{fig:AShNormLDstab}
\end{figure}

\emph{Appendix D: Fidelity estimation beyond Haar-random pure states}---The main text establishes concentration for Haar-random target states. Here we show that constant mean squared shadow norms also hold for other physically motivated ensembles.
\begin{proposition}\label{pro:AShNormFidUB}
Suppose $\caE$ is a state 2-design. Then:
\emph{(i)} For the uniform random-phase (URP) ensemble with respect to any orthonormal basis,
\begin{align}
\underset{\psi\sim\urp}{\bbE}\left\|\psi_0\right\|_\caE^2 \le\sqrt{\frac{48(d+1)(d+2)(d+3)}{d^3}}-1. \label{eq:AShNormRP}
\end{align}
\emph{(ii)} For $d=2^n$ and the $n$-qubit Clifford group $\Cl(n)$,
\begin{align}
\underset{U\sim\Cl(n)}{\bbE} \left\|U\psi_0 U^\dagger\right\|_\caE^2 \le 2\sqrt{15}-1 \quad \forall\,\psi\in\caP(\caH). \label{eq:AShNormFidUBCli}
\end{align}
\end{proposition}
\emph{URP ensembles}---The URP ensemble models long-time outputs of Hamiltonian evolution with random energy eigenvalues under a no-resonance condition \cite{Mark24}, a standard setting in quantum thermalization. The bound~\eref{eq:AShNormRP} holds for the URP ensemble with respect to any orthonormal basis and converges to $\sqrt{48}-1\approx 5.93$ as $d\to\infty$, guaranteeing efficient fidelity estimation with any 2-design measurement regardless of system size.

\emph{Clifford orbits}---The bound~\eref{eq:AShNormFidUBCli} is uniform over all initial states~$\psi$ and equals $2\sqrt{15}-1\approx 6.75$. It applies to any Clifford orbit, including orbits of stabilizer states, central to quantum error correction, and magic states, relevant to magic state distillation. Since Clifford orbits arise naturally in verification and benchmarking, efficient fidelity estimation for these targets is particularly valuable in practice.

\emph{Comparison}---Both bounds are slightly larger than the Haar-random mean of $2\sqrt{2}+1\approx 3.83$, reflecting less randomness in these ensembles. Nevertheless, all three bounds are $\caO(1)$, confirming that constant sample complexity for fidelity estimation is a robust structural feature of 2-design measurements, not an artifact of Haar randomness.

\clearpage
\newpage

\setcounter{equation}{0}
\setcounter{figure}{0}
\setcounter{table}{0}
\setcounter{theorem}{0}
\setcounter{proposition}{0}
\setcounter{lemma}{0}
\setcounter{section}{0}
\setcounter{page}{1}

\renewcommand{\theequation}{S\arabic{equation}}
\renewcommand{\thefigure}{S\arabic{figure}}
\renewcommand{\thetable}{S\arabic{table}}
\renewcommand{\thetheorem}{S\arabic{theorem}}
\renewcommand{\thelemma}{S\arabic{lemma}}
\renewcommand{\theproposition}{S\arabic{proposition}}
\renewcommand{\thecorollary}{S\arabic{corollary}}
\renewcommand{\thesection}{S\arabic{section}}
\renewcommand{\theHfigure}{S\arabic{figure}}
\renewcommand{\theHtable}{S\arabic{table}}
\renewcommand{\theHequation}{S\arabic{equation}}
\renewcommand{\theHproposition}{S\arabic{proposition}}

\onecolumngrid	
\begin{center}
	\textbf{\large Optimal Shadow Estimation with Minimal Measurement Settings: Supplemental Material}
\end{center}

\tableofcontents

\bigskip

In this Supplemental Material (SM), we prove the results presented in the main text and End Matter, including Propositions~\ref{pro:FP3UB}--\ref{pro:AShNormFidUB} as well as  \thsref{thm:MMSoptWCMUB}{thm:AvgShNormUB}. In addition, we provide some auxiliary results on frame potentials, the combined phase-design ensemble based on MUBs (optimal for worst-case shadow estimation), and shadow norms.

As in the main text, let $\caH$ be a complex Hilbert space of dimension $d\ge2$. Denote by $\caL(\caH)$, $\caL^{\rmH}(\caH)$, $\caL_0(\caH)$, and $\caL^{\rmH}_0(\caH)$ the spaces of linear operators, Hermitian operators, traceless operators, and traceless Hermitian operators on $\caH$, respectively. Denote by $\caD(\caH)$ the set of density operators on $\caH$ and by $\caP(\caH)$ the subset of pure state projectors. For any operator $O \in \caL(\caH)$, we write $O_0=O-\tr(O)\bbone/d$ for the traceless part of $O$; for any pure state $\ket{\phi} \in \caH$, we write $\phi = \proj{\phi}$ for the corresponding projector and $\phi_0=\phi-\bbone/d$. Given a positive integer $t$, denote by $P_{[t]}$ the projector onto the $t$-partite symmetric subspace of $\caH^{\otimes t}$ and by  $D_{[t]}=\tr(P_{[t]})$ the dimension of the symmetric subspace.

We also introduce the vectorization map between operators on $\caH$ and vectors in $\caH \otimes \caH$. For any $O \in \caL(\caH)$, its vectorization is defined in the computational basis $\{\ket{k}\}_{k=0}^{d-1}$ as
\begin{align}
\dket{O} \coloneqq \sum_{k,l} O_{kl}\ket{k}\otimes \ket{l},
\end{align}
where $O_{kl}=\bra{k}O\ket{l}$.

\section{\texorpdfstring{State $t$-designs and phase $t$-designs}{State t-designs and phase t-designs}} \label{sup:tDesignPhaseDesign}

Here we discuss state $t$-designs and phase $t$-designs in more detail and derive auxiliary results relevant to this work.

Let $\caE = \{\ket{\phi_i}, w_i\}_{i}$ be an ensemble of pure states in $\caH$, where  $w_i>0$ and $\sum_i w_i=1$. The $t$-th moment operator and normalized moment operator of $\caE$ are defined as
\begin{equation}
Q_t(\caE)\coloneqq\sum_i  w_i\phi_i^{\otimes t}, \quad \bQ_t(\caE)\coloneqq D_{[t]}Q_t(\caE).
\end{equation}
The $t$-th frame potential and normalized $t$-th frame potential of $\caE$ are defined as [see \eref{eq:FPt}]
\begin{equation}\label{eq:FPtSupp}
\Phi_t(\caE) \coloneqq \sum_{i,j} w_i w_j \bigl[\tr(\phi_i\phi_j)\bigr]^t,\quad \bPhi_t(\caE) \coloneqq D_{[t]}\,\Phi_t(\caE),
\end{equation}
which can also be expressed as
\begin{equation}
\Phi_t(\caE)=\tr\left\{[Q_t(\caE)]^2\right\},\quad \bPhi_t(\caE)=D_{[t]}\tr\left\{[Q_t(\caE)]^2\right\}=\frac{\tr\left\{[\bQ_t(\caE)]^2\right\}}{D_{[t]}}.
\end{equation}
By definition, we have
\begin{equation} \label{eq:Phittp1}
   \Phi_{t+1}(\caE) \leq \Phi_t(\caE),\quad \bPhi_{t+1}(\caE)\leq \frac{d+t}{t+1} \bPhi_t(\caE),
\end{equation}
given that $D_{[t+1]}=(d+t)D_{[t]}/(t+1)$.

\subsection{\texorpdfstring{Haar-random states and state $t$-designs}{Haar-random states and state t-designs}}

State $t$-designs, also known as complex projective $t$-designs, are configurations of pure states whose $t$-th moments match those of Haar-random states \cite{Zauner11,Renes04,Scott06,Ambainis07}. They can be viewed as the complex analog of spherical designs on the real unit sphere and have found applications in areas such as approximation theory, combinatorics, and quantum information. The ensemble $\caE$ is a state $t$-design if $\bQ_t(\caE)=P_{[t]}$. It is known that $\bPhi_t(\caE)\geq 1$, and the lower bound is saturated if and only if $\caE$ is a $t$-design. In addition, any state $t$-design in dimension $d$ has at least
\begin{align}
\binom{d+\lceil t/2\rceil-1}{\lceil t/2\rceil} \binom{d+\lfloor t/2\rfloor-1}{\lfloor t/2\rfloor}
\end{align}
elements \cite{Hoggar82,Scott06}; the lower bound simplifies to $d^2$ and $d^2(d+1)/2$ for $t=2$ and $t=3$, respectively. When $t=2$, the lower bound is saturated if and only if the ensemble $\caE$ induces a SIC-POVM. It is conjectured and supported by strong evidence that a SIC exists in every finite dimension \cite{Fuchs17,Zauner11,Renes04,Scott10}.
In sharp contrast, it remains open whether a 3-design can always be constructed using only $\caO(d^3)$ states. As a prominent example, the set of $n$-qubit stabilizer states forms a 3-design~\cite{Kueng15,Zhu17,Webb16}, yet the number of such states grows superpolynomially with dimension~$d$.

\subsection{\texorpdfstring{Uniform random phase states and phase $t$-designs}{Uniform random phase states and phase t-designs}}
\label{app:RPtd}
Given a positive integer $t$, two sequences $\bmk,\bml \in \{0,1,\dots,d-1\}^t$ are said to be equivalent if one can be obtained from the other by a permutation of entries. The equivalence class of $\bmk$ is denoted by $[\bmk]$, and its cardinality by $|[\bmk]|$. Let $\caB = \{\ket{k}_\caB\}_{k=0}^{d-1}$ be an orthonormal basis of $\caH$. For each equivalence class $[\bmk]$, we define the corresponding symmetric ket in $\caH^{\otimes t}$ by
\begin{align}
\ket{[\bmk]}_\caB \coloneqq \frac{1}{\sqrt{|[\bmk]|}} \sum_{\bml \in [\bmk]} \ket{l_1}_\caB \otimes \cdots \otimes \ket{l_t}_\caB.
\end{align}
Recall that the uniform random phase (URP) ensemble $\caT_\caB$ with respect to $\caB=\{\ket{k}_\caB\}$ has the form
\begin{equation}
    \caT_\caB=\left\{\frac{1}{\sqrt{d}} \sum_{k=0}^{d-1} \rme^{\rmi \varphi_k} \ket{k}_\caB\colon \rme^{\rmi\varphi_k}\sim \rmU(1)\right\}, \label{eq:URPsupp}
\end{equation}
where $\varphi = (\varphi_0, \varphi_1, \dots, \varphi_{d-1})$ is uniformly distributed over $[0,2\pi)^d$. When $\caB$ is the computational basis, we omit the subscript $\caB$ and simply write $\caT$.

% When no confusion arises, we omit the subscript $\caB$ and simply write $\caT$.

\begin{lemma}\label{lem:URPAVG}
For the URP ensemble $\caT_\caB$, we have
\begin{align}
Q_t(\caT_\caB) = \frac{1}{d^t} \sum_{[\bm{k}]} |[\bm{k}]| \proj{[\bm{k}]}_\caB\le \frac{t!}{d^t}P_{[t]},\quad t=1,2,\ldots \label{eq:URP_AVG}
\end{align}
\end{lemma}

We now define phase designs. Let $\caE_\caB$ be an ensemble of phase states (equal-modulus states) of the form
\begin{align}\label{eq:PhaseState}
\ket{\phi} = \frac{1}{\sqrt d} \sum_{k=0}^{d-1}\rme^{\rmi\varphi_k}\ket{k}_\caB.
\end{align}
The ensemble $\caE_\caB$ is called a phase $t$-design with respect to $\caB$ if
\begin{align}
Q_t(\caE_\caB) = Q_t(\caT_\caB).
\end{align}
By definition, any phase $t$-design with $t\geq 1$ is also a state 1-design and thus defines a POVM. In addition, a phase $t$-design determines a projective toric $t$-design \cite{Iosue24}, and vice versa.

\begin{proof}[Proof of \lref{lem:URPAVG}]
By definition, a state drawn from the URP ensemble has the form in \eref{eq:PhaseState} with $\rme^{\rmi\varphi_k}\sim \rmU(1)$.
The $t$-th tensor power of the corresponding projector reads
\begin{align}
\phi^{\otimes t} = \left(\proj{\phi}\right)^{\otimes t} = \frac{1}{d^t} \sum_{k_1,\dots,k_t} \sum_{l_1,\dots,l_t} \rme^{\rmi(\varphi_{k_1}+\dots+\varphi_{k_t}-\varphi_{l_1}-\dots-\varphi_{l_t})} \ket{k_1\cdots k_t}\bra{l_1\cdots l_t}_\caB.
\end{align}
The expectation of the phase factor over $\caT_\caB$ reads
\begin{align}
\underset{\phi\sim\caT_\caB}{\bbE} \left[\rme^{\rmi(\sum_{s=1}^t \varphi_{k_s} - \sum_{s=1}^t \varphi_{l_s})}\right] = \begin{cases} 1 & \text{if } [\bmk]=[\bml], \\ 0 & \text{otherwise}. \end{cases}
\end{align}
The above two equations together imply the equality in \eref{eq:URP_AVG}. The inequality holds because $\left\{\ket{[\bmk]}_\caB\right\}_{[\bmk]}$ forms an orthonormal basis of the symmetric subspace of $\caH^{\otimes t}$
and $|[\bmk]|\leq t!$.
\end{proof}

\subsection{Proof of \pref{pro:FP3UB}}
\begin{proof}[Proof of \pref{pro:FP3UB}]
Suppose the state ensemble $\caE$ has the form $\caE=\{\ket{\phi_i},w_i\}_i$ and let $\phi_{i,0}=\phi_i-\bbone/{d}$ for each index $i$. Since $0\leq \tr(\phi_i\phi_j)\leq 1$ for all $i,j$, we have
\begin{align}
\sum_{i,j}w_iw_j\tr(\phi_i\phi_j)\left[\tr\left(\phi_{i,0}\phi_{j,0}\right)\right]^2\le \sum_{i,j}w_iw_j\left[\tr\left(\phi_{i,0}\phi_{j,0}\right)\right]^2. \label{eq:phi_3_eq1}
\end{align}
Direct calculation yields
\begin{align}
\sum_{i,j}w_iw_j\tr(\phi_i\phi_j)\left[\tr\left(\phi_{i,0}\phi_{j,0}\right)\right]^2&=\Phi_3(\caE)-\frac{2}{d}\Phi_2(\caE)+\frac{1}{d^2}\Phi_1(\caE)=\Phi_3(\caE)-\frac{3d-1}{d^3(d+1)},\\
\sum_{i,j}w_iw_j\left[\tr\left(\phi_{i,0}\phi_{j,0}\right)\right]^2&=\Phi_2(\caE)-\frac{2}{d}\Phi_1(\caE)+\frac{1}{d^2}=\frac{d-1}{d^2(d+1)},
\end{align}
where the second equality in each line holds because $\Phi_1(\caE)=1/d$ and $\Phi_2(\caE)=2/[d(d+1)]$, given that $\caE$ is a 2-design. Combining the above three equations yields 
\begin{equation}
  \Phi_3(\caE)\leq  \frac{d^2+2d-1}{d^3(d+1)},\quad \bPhi_3(\caE)
		\leq \frac{(d+2)(d^2+2d-1)}{6d^2}. 
\end{equation}
which confirms
\pref{pro:FP3UB}.
\end{proof}

\subsection{Proofs of \psref{pro:Phase2design}{pro:Phase3design}}
\Pref{pro:Phase3design} is a simple corollary of the following lemma, which guarantees that the third moment operators of $\caT_3$ and $\tcaT_3$ have the form in \lref{lem:URPAVG}. \Pref{pro:Phase2design} follows from similar but simpler reasoning.

\begin{lemma}\label{lem:DesignPermutation}
Suppose $x_1, x_2, x_3, y_1, y_2, y_3\in\{0,1,2,\ldots, d-1\}$. Then the sequence $y_1, y_2, y_3$ is a permutation of $x_1, x_2, x_3$ if and only if
\begin{equation}\label{eq:ModularCondition}
\begin{aligned}
     x_1+x_2+x_3&=y_1+y_2+y_3\mmod d,\\
f_3(x_1)+f_3(x_2)+f_3(x_3)&=f_3(y_1)+f_3(y_2)+f_3(y_3)\mmod 7,\\
x_1^2+x_2^2+x_3^2&=y_1^2+y_2^2+y_3^2\mmod p,\\
x_1^3+x_2^3+x_3^3&=y_1^3+y_2^3+y_3^3\mmod p.
\end{aligned}
\end{equation}
If $d$ is not divisible by 3, then the sequence $y_1, y_2, y_3$ is a permutation of $x_1, x_2, x_3$ if and only if
\begin{equation}\label{eq:ModularCondition2}
\begin{aligned}
     x_1+x_2+x_3&=y_1+y_2+y_3\mmod d,\\
x_1+x_2+x_3&=y_1+y_2+y_3\mmod 3,\\
x_1^2+x_2^2+x_3^2&=y_1^2+y_2^2+y_3^2\mmod p,\\
x_1^3+x_2^3+x_3^3&=y_1^3+y_2^3+y_3^3\mmod p.
\end{aligned}
\end{equation}
\end{lemma}

\begin{proof}[Proof of \lref{lem:DesignPermutation}]
If the sequence $y_1, y_2, y_3$ is a permutation of $x_1, x_2, x_3$, then \eref{eq:ModularCondition} holds automatically. Conversely, suppose \eref{eq:ModularCondition} holds. Then the second line implies that 
\begin{equation}
f_3(x_1)+f_3(x_2)+f_3(x_3)=f_3(y_1)+f_3(y_2)+f_3(y_3),\quad |y_1+y_2+y_3-x_1-x_2-x_3|<d,
\end{equation}
given that $0\leq f_3(x)\leq 2$ for $x\in \{0,1,2,\ldots, d-1\}$. The first two lines in \eref{eq:ModularCondition} together yield $x_1+x_2+x_3=y_1+y_2+y_3$ and
\begin{equation}\label{eq:ElementaryPoly1}
x_1+x_2+x_3=y_1+y_2+y_3\mmod p.
\end{equation}
In conjunction with the last two lines in \eref{eq:ModularCondition}, we can deduce that
\begin{equation}\label{eq:ElementaryPoly23}
\begin{aligned}
x_1 x_2 +x_2 x_3+x_3x_1 &= y_1 y_2 +y_2 y_3+y_3y_1 \mmod p, \\
x_1 x_2 x_3 &= y_1 y_2 y_3 \mmod p.
\end{aligned}
\end{equation}
\Eqsref{eq:ElementaryPoly1}{eq:ElementaryPoly23} together imply that the sequence $y_1, y_2, y_3$ is a permutation of $x_1, x_2, x_3$.

Next, suppose $d$ is not divisible by 3; then $d$ and 3 are coprime. If the sequence $y_1, y_2, y_3$ is a permutation of $x_1, x_2, x_3$, then \eref{eq:ModularCondition2} holds automatically. Conversely, suppose \eref{eq:ModularCondition2} holds. The first two lines together give $x_1+x_2+x_3=y_1+y_2+y_3 \mmod 3d$, which in turn implies that $x_1+x_2+x_3=y_1+y_2+y_3$. So \eqsref{eq:ElementaryPoly1}{eq:ElementaryPoly23} hold as before, and the sequence $y_1, y_2, y_3$ is a permutation of $x_1, x_2, x_3$.
\end{proof}

\subsection{Auxiliary results on frame potentials}
\label{app:potential}
Here we introduce auxiliary results on (normalized) frame potentials and their variants, which are useful for studying shadow norms in shadow estimation.

Suppose $\caE=\{\ket{\phi_{i}}, w_i\}_{i}$ is a pure state ensemble on $\caH$ and $\psi \in \caP(\caH)$. We define the $t$-th frame potential of $\caE$ relative to $\psi$ and its normalized version as
\begin{align}  \Phi_t(\caE,\psi)\coloneqq\sum_iw_i\left[\tr(\psi \phi_i)\right]^t,\quad \bPhi_t(\caE,\psi)\coloneqq D_{[t]}\Phi_t(\caE,\psi).
\end{align}
where $D_{[t]}=\binom{d + t - 1}{t}$ is the dimension of the symmetric subspace  of $\caH^{\otimes t}$. These definitions are also applicable when $t=0$, in which case 
$\bPhi_t(\caE,\psi)=\Phi_t(\caE,\psi)=\bPhi_t(\caE)=\Phi_t(\caE)=D_{[t]}=1$. 
By definition, $\Phi_t(\caE)=\sum_j w_j\Phi_t(\caE,\phi_j)$, which implies
\begin{equation}
\min_{\psi\in \caP(\caH)} \Phi_t(\caE,\psi) \le \Phi_t(\caE) \le \max_{\psi\in \caP(\caH)} \Phi_t(\caE,\psi).
\end{equation}

% For the convenient of the following discussion we also define 

\begin{lemma}\label{lem:FPrelation}
Suppose $\caE$ is a state ensemble on $\caH$, $\psi\in\caP(\caH)$, and $t$ is a positive integer. Then
\begin{align}
   \Phi_t^2(\caE)&\leq \Phi_{t-1}(\caE) \Phi_{t+1}(\caE), & \Phi_t^2(\caE,\psi)&\leq \Phi_{t-1}(\caE,\psi) \Phi_{t+1}(\caE,\psi), \label{eq:FPrelation}\\
     \frac{\bPhi_t^2(\caE)}{D_{[t]}^2}&\leq \frac{\bPhi_{t-1}(\caE) \bPhi_{t+1}(\caE)}{D_{[t-1]}D_{[t+1]}},& \frac{\bPhi_t^2(\caE,\psi)}{D_{[t]}^2}&\leq \frac{\bPhi_{t-1}(\caE,\psi) \bPhi_{t+1}(\caE,\psi)}{D_{[t-1]}D_{[t+1]}}.
     \label{eq:NFPrelation}
\end{align}
If $\caE$ forms a state 2-design, then
\begin{align}
     \bPhi_3(\caE,\psi)\ge\dfrac{2(d+2)}{3(d+1)}> \dfrac{2}{3}, \quad \bPhi_4(\caE,\psi)\ge\dfrac{(d+2)(d+3)}{3(d+1)^2}> \dfrac{1}{3}, \label{eq:NFP34LB}
\end{align}
and the same bounds hold with $\bPhi_t(\caE,\psi)$ replaced by $\bPhi_t(\caE)$ for $t=3,4$.
\end{lemma}

\begin{lemma}\label{lem:8th_moment}
Suppose $\psi,\phi\in\caP(\caH)$. Then
\begin{align}
   \tr\left[P_{[8]}\left(\psi^{\otimes 4}\otimes \phi^{\otimes 4}\right)\right]=\frac{1+16\tr(\psi\phi)+36\left[\tr(\psi\phi)\right]^2+16\left[\tr(\psi\phi)\right]^3+\left[\tr(\psi\phi)\right]^4}{70}.
\end{align}
\end{lemma}

\begin{lemma}\label{lem:4th_poten}
Suppose $\caE$ is a state ensemble on $\caH$ and $\psi \in \caP(\caH)$ is a Haar-random pure state. Then
\begin{align}
     \underset{\psi\sim \Ha}{\bbE} \bPhi_4(\caE,\psi)&=1, \label{eq:4th_moment}\\
   \underset{\psi\sim \Ha}{\bbE} \left[\bPhi_4(\caE,\psi)\right]^2&=\frac{D_1}{D_2}\left(1+\dfrac{16\bPhi_1(\caE)}{D_{[1]}}+\dfrac{36\bPhi_2(\caE)}{D_{[2]}}+\dfrac{16\bPhi_3(\caE)}{D_{[3]}}+\dfrac{\bPhi_4(\caE)}{D_{[4]}}\right), \label{eq:4th_squared}
\end{align}
where $D_1=d(d+1)(d+2)(d+3)$ and $D_2=(d+4)(d+5)(d+6)(d+7)$. If $\caE$ forms a state 2-design, then
\begin{align}
\Var(\bPhi_4(\caE,\psi))&\le\frac{6(17d+51)[\bPhi_3(\caE)-1]+6(d-1)}{D_2}<\frac{[\xi(\caE)]^2}{(d+6)^3}, \label{eq:4th_mo_var}
\end{align}
where $\xi(\caE)=\sqrt{102[\bPhi_3(\caE)-1]+6}$, as defined in \eref{eq:Xi}.
\end{lemma}

\begin{lemma}\label{lem:RP_URP}
Suppose $\caE$ is a state ensemble on $\caH$ and $\caB$ is an orthonormal basis of $\caH$. Then
\begin{align}
     \underset{\psi\sim \caT_\caB}{\bbE}\bPhi_t(\caE, \psi)\le\frac{t!D_{[t]}}{d^t},\label{eq:RP_URP}
\end{align}
where $\caT_\caB$ is the uniform random phase ensemble with respect to $\caB$, as defined in \eref{eq:URPsupp}.
\end{lemma}

\subsection{Proofs of auxiliary results \lsref{lem:FPrelation}{lem:RP_URP}}

\begin{proof}[Proof of \lref{lem:FPrelation}]
\Eref{eq:FPrelation} follows from the Cauchy--Schwarz inequality applied to the definitions of frame potentials and relative frame potentials. \Eref{eq:NFPrelation} is a simple corollary of \eref{eq:FPrelation} given that $\bPhi_t(\caE)=D_{[t]}\Phi_t(\caE)$ and $\bPhi_t(\caE,\psi)=D_{[t]}\Phi_t(\caE,\psi)$ for any positive integer $t$.

Next, suppose $\caE$ is a 2-design; then $\bPhi_1(\caE,\psi)=\bPhi_1(\caE)=1$ and $\bPhi_2(\caE,\psi)=\bPhi_2(\caE)=1$. By virtue of \eref{eq:NFPrelation} with $t=2,3$, we deduce that
\begin{align}
\bPhi_3(\caE,\psi)&\ge \frac{D_{[1]}D_{[3]}\bPhi_2^2(\caE,\psi)}{D_{[2]}^2 \bPhi_1(\caE,\psi)}=\frac{D_{[1]}D_{[3]}}{D_{[2]}^2}=\dfrac{2(d+2)}{3(d+1)}> \dfrac{2}{3},\\
\bPhi_4(\caE,\psi)&\ge \frac{D_{[2]}D_{[4]}\bPhi_3^2(\caE,\psi)}{D_{[3]}^2 \bPhi_2(\caE,\psi)} =\frac{3(d+3)\bPhi_3^2(\caE,\psi)}{4(d+2)} \ge\dfrac{(d+2)(d+3)}{3(d+1)^2}> \dfrac{1}{3},
\end{align}
which confirm \eref{eq:NFP34LB}. The same reasoning applies with $\bPhi_t(\caE,\psi)$ replaced by $\bPhi_t(\caE)$ for $t=1,2,3,4$. This completes the proof of \lref{lem:FPrelation}.
\end{proof}

\begin{proof}[Proof of \lref{lem:8th_moment}]
Note that $P_{[8]}=\sum_{\sigma \in S_8} R(\sigma)/8!$, where $S_t$ for a positive integer $t$ denotes the symmetric group of $t$ numbers, and $R(\sigma)$ denotes the unitary representation of the permutation $\sigma$. Accordingly,
\begin{align}
\tr\left[P_{[8]}\left(\psi^{\otimes 4}\otimes \phi^{\otimes 4}\right)\right] = \frac{1}{8!} \sum_{\sigma \in S_8} \tr\left[R(\sigma)\left(\psi^{\otimes 4}\otimes \phi^{\otimes 4} \right)\right]= \frac{1}{8!} \sum_{\sigma \in S_8} [\tr(\psi\phi)]^{\gamma(\sigma)},
\end{align}
where $\gamma(\sigma)=|\sigma(\{1,2,3,4\})\cap\{5,6,7,8\}|$. A straightforward counting argument shows that
\begin{align}
|\{\sigma \in S_8 \mid \gamma(\sigma)=m \}|=\begin{cases} 576 & m=0,\, 4,\\ 9216 & m=1,\, 3,\\ 20736 & m=2. \end{cases}
\end{align}
Combining the above two equations, we obtain
\begin{align}
\tr\left[P_{[8]}\left(\psi^{\otimes 4}\otimes \phi^{\otimes 4}\right)\right]&= \frac{576+9216\tr(\psi\phi)+20736\left[\tr(\psi\phi)\right]^2+9216\left[\tr(\psi\phi)\right]^3+576\left[\tr(\psi\phi)\right]^4}{8!}\nonumber\\
&=\frac{1+16\tr(\psi\phi)+36\left[\tr(\psi\phi)\right]^2+16\left[\tr(\psi\phi)\right]^3+\left[\tr(\psi\phi)\right]^4}{70}, \label{eq:8th_moment}
\end{align}
which completes the proof of \lref{lem:8th_moment}.
\end{proof}

\begin{proof}[Proof of \lref{lem:4th_poten}]
For any positive integer $t$, we have
\begin{align}
    \underset{\psi\sim \Ha}{\bbE}\psi^{\otimes t}=\dfrac{P_{[t]}}{D_{[t]}},
\end{align}
where $P_{[t]}$ denotes the projector onto the $t$-partite symmetric subspace of $\caH^{\otimes t}$ and $D_{[t]}=\tr(P_{[t]})=\binom{d + t - 1}{t}$. Based on this observation, \eref{eq:4th_moment} can be proved as follows:
\begin{align}
\underset{\psi\sim \Ha}{\bbE} \bPhi_4(\caE,\psi)=D_{[4]}\underset{\psi\sim \Ha}{\bbE}\sum_iw_i\left[\tr(\psi\phi_i)\right]^4=D_{[4]}\sum_iw_i\tr\left[\left(\underset{\psi\sim \Ha}{\bbE} \psi^{\otimes 4}\right)\phi_i^{\otimes 4}\right]=\sum_iw_i\tr\left(P_{[4]}\phi_i^{\otimes 4}\right)=1.
\end{align}
In conjunction with \lref{lem:8th_moment}, \eref{eq:4th_squared} can be verified as follows:
\begin{align}
    &\underset{\psi\sim \Ha}{\bbE} \left[\bPhi_4(\caE,\psi)\right]^2=D_{[4]}^2\sum_{i,j}w_iw_j\underset{\psi\sim \Ha}{\bbE}\left[\left[\tr(\psi \phi_{i})\right]^4\left[\tr(\psi \phi_{j})\right]^4\right] = \dfrac{D_{[4]}^2}{D_{[8]}}\sum_{i,j}w_iw_j\tr\left[P_{[8]}\left(\phi_{i}^{\otimes 4}\otimes \phi_{j}^{\otimes 4}\right)\right]\nonumber\\
    &=\dfrac{D_{[4]}^2}{D_{[8]}}\sum_{i,j}\frac{w_iw_j\left\{1+16\tr(\phi_i\phi_j)+36\left[\tr(\phi_i\phi_j)\right]^2+16\left[\tr(\phi_i\phi_j)\right]^3+\left[\tr(\phi_i\phi_j)\right]^4\right\}}{70}\nonumber\\
&=\frac{D_{[4]}^2}{70D_{[8]}}\left[1+16\Phi_1(\caE)+36\Phi_2(\caE)+16\Phi_3(\caE)+\Phi_4(\caE)\right]=\frac{D_1}{D_2}\left(1+\dfrac{16\bPhi_1(\caE)}{D_{[1]}}+\dfrac{36\bPhi_2(\caE)}{D_{[2]}}+\dfrac{16\bPhi_3(\caE)}{D_{[3]}}+\dfrac{\bPhi_4(\caE)}{D_{[4]}}\right),
\end{align}
where $D_1=d(d+1)(d+2)(d+3)$ and $D_2=(d+4)(d+5)(d+6)(d+7)$.

Next, suppose $\caE$ is a state 2-design; then $\bPhi_2(\caE)=\bPhi_1(\caE)=1$. The variance of $\bPhi_4(\caE,\psi)$ can be upper bounded as follows:
\begin{align}
\Var(\bPhi_4(\caE,\psi)) &= \underset{\psi\sim \Ha}{\bbE} \left[\bPhi_4(\caE,\psi)\right]^2-\left[\underset{\psi\sim \Ha}{\bbE}\bPhi_4(\caE,\psi)\right]^2=\frac{D_1}{D_2}\left(1+\dfrac{16\bPhi_1(\caE)}{D_{[1]}}+\dfrac{36\bPhi_2(\caE)}{D_{[2]}}+\dfrac{16\bPhi_3(\caE)}{D_{[3]}}+\dfrac{\bPhi_4(\caE)}{D_{[4]}}\right)-1\nonumber\\
&=\frac{24[(4d+12)\bPhi_3(\caE)+\bPhi_4(\caE)-4d-13]}{D_2}\le \frac{6(17d+51)[\bPhi_3(\caE)-1]+6(d-1)}{D_2}\nonumber\\
&<\frac{(d+3)\{102[\bPhi_3(\caE)-1]+6\}}{D_2} <\frac{[\xi(\caE)]^2}{(d+6)^3},
\end{align}
which confirms \eref{eq:4th_mo_var}. Here the first inequality holds because $\bPhi_4(\caE)\le {(d+3)\bPhi_3(\caE)}/{4}$ by \eref{eq:Phittp1}, the second inequality holds because
\begin{equation}
\bPhi_3(\caE)\ge 1,\quad \frac{(17d+51)[\bPhi_3(\caE)-1]+d-1}{d+3}<17[\bPhi_3(\caE)-1]+1,
\end{equation}
and the last inequality holds because $(d+3)/D_2<1/(d+6)^3$. This completes the proof of \lref{lem:4th_poten}.
\end{proof}

\begin{proof}[Proof of \lref{lem:RP_URP}]
Suppose the ensemble $\caE$ has the form $\caE=\{\ket{\phi_i},w_i\}_i$. Then \eref{eq:RP_URP} can be proved as follows:
\begin{align}
     \underset{\psi\sim \caT_\caB}{\bbE}\bPhi_t(\caE, \psi)&=D_{[t]}\sum_iw_i\underset{\psi\sim \caT_\caB}{\bbE}\left[\tr(\phi_i\psi)\right]^t=D_{[t]}\sum_iw_i\tr\left[\phi_i^{\otimes t}\left( \underset{\psi\sim \caT_\caB}{\bbE}\psi^{\otimes t}\right)\right]\nonumber\\&\le \frac{t! D_{[t]}}{d^t}\sum_iw_i\tr\left(\lsp\phi_i^{\otimes t}P_{[t]}\right)=\frac{t! D_{[t]}}{d^t},
\end{align}
where the inequality follows from \lref{lem:URPAVG}.
\end{proof}

\section{Shadow estimation within the framework of rank-1 IC-POVMs}
\label{sup:shadow_esti}

In this section we briefly introduce shadow estimation within the framework of rank-1 IC POVMs, following the presentation in the main text.
Let $\caE = \{\ket{\phi_i}, w_i\}_i$ be a state ensemble that forms a $1$-design, which induces the POVM $\{ d w_i \phi_i \}_i$. If $\caE$ is informationally complete (IC), any quantum state $\rho \in \caD(\caH)$ can be reconstructed as $\rho = d \sum_i w_i \tr(\phi_i \rho)\, \hat{\rho}_i$, where $\hat{\rho}_i$ is the estimator associated with outcome $i$. When $\caE$ is informationally overcomplete, the choice of $\hat{\rho}_i$ is not unique. In the shadow-estimation framework, one usually employs the canonical reconstruction, which is the standard choice in the absence of prior information about $\rho$.

The canonical measurement channel can be represented in two equivalent ways. The first directly defines the channel action:
\begin{align}
\caM_\caE(\rho) &= d \sum_i w_i \tr(\phi_i \rho)\, \phi_i. \label{eq:FD1}
\end{align}
The second uses the vectorization (double-ket) notation. For any operator $A \in \caL^{\rmH}(\caH)$, we define its vectorization $\dket{A} \in \caH \otimes \caH$ by expanding $\dket{A} = \sum_{k,l} A_{kl}\, \ket{k}\otimes\ket{l}$ in a fixed orthonormal basis, so that the Hilbert--Schmidt inner product becomes $\dbraket{A}{B}=\tr\left(A^\dagger B\right)$. The outer product $\dproj{A} = \dket{A}\dbra{A}$ then acts as a linear operator on this space. With this notation the measurement channel becomes
\begin{align}
\caM_\caE = d \sum_i w_i \dproj{\phi_i}, \label{eq:FD2}
\end{align}
where $\caM_\caE$ is viewed as an operator acting on vectorized operators via $\caM_\caE\dket{A} = \dket{\caM_\caE(A)}$. This representation makes the positive semidefiniteness of $\caM_\caE$ manifest and facilitates spectral analysis and inversion of the channel.
By definition, $\caM_\caE(\bbone) = \bbone$, so $\bbone$ is an eigenoperator of $\caM_\caE$ with eigenvalue $1$. With a slight abuse of notation, we use $\caM_\caE$ to denote both representations of the measurement channel when no confusion arises. 
Let $\bbfI$ denote the projector onto the space of traceless Hermitian operators. The restricted channel is
\begin{align}
    \bcaM_\caE \coloneqq \bbfI\, \caM_\caE\, \bbfI = d \sum_i w_i \dproj{\phi_{i,0}},
\end{align}
where $\phi_{i,0} = \phi_i - \bbone/d$ is the traceless part of $\phi_i$. Because $\caE$ forms a $1$-design, this simplifies to
\begin{align}
    \bcaM_\caE = \caM_\caE - \frac{1}{d}\dproj{\bbone}.
\end{align}

For any observable $O \in \caL^{\rmH}(\caH)$ and an IC-POVM $\caM = \{E_i\}_i$, the expectation value $\tr(\rho O)$ can be estimated as
\begin{align}
    \tr(\rho O) = \sum_i \tr(E_i \rho)\, \hat{o}_i,
\end{align}
where $\hat{o}_i = \tr(\hat{\rho}_i O)$ and $\hat{\rho}_i$ is the canonical reconstruction operator for outcome $i$. The single-shot estimation variance reads
\begin{align}
    \operatorname{Var}(\hat{o}) = \sum_i \tr(E_i \rho)\, \hat{o}_i^2 - [\tr(\rho O)]^2 = \sum_i \tr(E_i \rho)\, \hat{o}_{0,i}^2 - [\tr(\rho O_0)]^2,
\end{align}
where $O_0 = O - \tr(O)\bbone/d $ is the traceless part of $O$ and $\hat{o}_{0,i} = \tr(\hat{\rho}_i O_0)$. The second equality shows that the variance depends solely on $O_0$. The dominant term defines the \emph{state-dependent squared shadow norm}:
\begin{align}
    \|O_0\|_{\caE,\rho}^2 \coloneqq \sum_i \tr(E_i \rho)\, \hat{o}_{0,i}^2 = d \sum_i w_i \tr(\rho \phi_i) \left\{ \tr\left[ \phi_i\, \caM_\caE^{-1}(O_0) \right] \right\}^2=d\tr\left\{Q_3(\caE) \left[\rho \otimes \caM_{\caE}^{-1}(O_0)\otimes \caM_{\caE}^{-1}(O_0)\right]\right\},  \label{eq:ShadowNormSD}
\end{align}
where $\caM_\caE^{-1}$ denotes the inverse of the measurement channel. The (state-independent) squared shadow norm is obtained by maximizing $\|O\|_{\caE,\rho}^2$ over all states:
\begin{align}
    \|O_0\|_\caE^2 \coloneqq \max_{\rho\in\caD(\caH)} \|O_0\|_{\caE,\rho}^2 = \left\| d \sum_i w_i \phi_i \left\{ \tr\left[ \phi_i\, \caM_\caE^{-1}(O_0) \right] \right\}^2 \right\|=\left\|d\tr_{2,3}\left\{Q_3(\caE) \left[\bbone \otimes \caM_{\caE}^{-1}(O_0)\otimes \caM_{\caE}^{-1}(O_0)\right]\right\}\right\|, \label{eq:norm2}
\end{align}
which applies to any traceless operator $O_0 \in \caL^{\rmH}_0(\caH)$.

If $\caE$ forms a state $2$-design, the measurement channel takes the simple form
\begin{align}
    \caM_\caE = \frac{\bfI + \dproj{\bbone}}{d+1},
\end{align}
where $\bfI$ is the identity superoperator on $\caL^{\rmH}(\caH)$. Restricted to traceless operators, its inverse gives
\begin{align}
    \caM_\caE^{-1}(O_0) = \bcaM_\caE^{-1}(O_0) = (d+1)\, O_0.
\end{align}
Consequently, the squared shadow norm reduces to
\begin{align}
    \|O_0\|_\caE^2 = \left\| d(d+1)^2 \sum_i w_i \phi_i \left[ \tr(\phi_i O_0) \right]^2 \right\|=\left\|d(d+1)^2\tr_{2,3}\left[Q_3(\caE) \left(\bbone \otimes O_0\otimes O_0\right)\right]\right\|.
\end{align}
If $\caE$ further forms a state $3$-design, then $Q_3(\caE)=6P_{[3]}/[d(d+1)(d+2)]$, and  we obtain the closed-form expression
\begin{align}
    \|O_0\|_\caE^2 = \frac{d+1}{d+2}\left( 2\|O_0\|^2 + \|O_0\|_2^2 \right) \le 3\|O_0\|_2^2.
\end{align}
For fidelity estimation, where the observable is a pure-state projector $\psi$ with traceless part $\psi_0 = \psi - \bbone/d $, the shadow norm evaluates to
\begin{align}
    \|\psi_0\|_\caE^2 = \frac{(3d-2)(d+1)(d-1)}{d^2(d+2)}. \label{eq:3design_pure}
\end{align}

\section{Proofs of results on worst-case optimal shadow estimation}

\subsection{Proof of \pref{pro:ShNormLUB}}
\label{sec:pro:ShNormLUBProof}
\begin{proof}
Suppose the ensemble has the form $\caE=\{\ket{\phi_i},w_i\}_{i=1}^{K}$. For any collection of states $\{\rho_j\}_{j=1}^K \subset \caD(\caH)$, any collection of rank-1 projectors $\{\psi_j\}_{j=1}^K \subset \caP(\caH)$, and any probability distribution $\{v_j\}_{j=1}^K$ (i.e., $v_j \ge 0$ and $\sum_j v_j = 1$), we have
\begin{align}
\max_{\psi\in \caP(\caH)} \|\psi_0\|_\caE^2
&\coloneqq \max_{\rho\in \caD(\caH),\,\psi\in \caP(\caH)} d\sum_i w_i\tr(\rho \phi_i) \left\{\tr\left[\phi_i\caM_\caE^{-1}(\psi_0) \right]\right\}^2 \nonumber\\
&\,\ge d\sum_{i,j} w_i v_j{\tr(\rho_j \phi_i) \left\{\tr\left[\phi_i\caM_\caE^{-1}(\psi_{j,0}) \right]\right\}^2}, \label{eq:LBK}
\end{align}
where the inequality follows from the fact that the maximum is lower bounded by any convex combination. To establish the inequality in \pref{pro:ShNormLUB}, we now choose
\begin{align}
\rho_j=\phi_j, \qquad \psi_{j}=\phi_j, \qquad v_j=\frac{1}{K}.
\end{align}
Then we obtain
\begin{align}
\max_{\psi\in \caP(\caH)} \|\psi_0\|_\caE^2
&\ge d\sum_{i,j} w_i v_j \tr(\rho_j \phi_i) \left\{\tr\left[\phi_i\caM_\caE^{-1}(\phi_{j,0}) \right]\right\}^2 \ge \frac{d}{K}\sum_i w_i \tr(\rho_i \phi_i) \left\{\tr\left[\phi_i\caM_\caE^{-1}(\phi_{i,0}) \right]\right\}^2 \nonumber\\
&= \frac{d}{K}\sum_i w_i \left\{\tr\left[\phi_{i,0}\bcaM_\caE^{-1}(\phi_{i,0}) \right]\right\}^2.\label{eq:LBIC}
\end{align}
The second inequality holds because $\tr(\rho_j\phi_i)$ and $\left\{\tr\left[\phi_i\caM_\caE^{-1}(\phi_{j,0}) \right]\right\}^2$ are nonnegative for all $i,j$; the equality follows from the facts that $\phi_{i,0}$ is traceless and that $\caM_\caE$ restricted to $\caL^{\rmH}_0(\caH)$ coincides with $\bcaM_\caE$.

Since $f(x)=x^t$ is convex for $t>1$, Jensen's inequality implies that
\begin{align}
\sum_i w_i \left\{\tr\left[\phi_{i,0}\bcaM_\caE^{-1}(\phi_{i,0}) \right]\right\}^2 \ge \left\{\sum_i w_i \tr\left[\phi_{i,0}\bcaM_\caE^{-1}(\phi_{i,0}) \right]\right\}^2.\label{eq:Jensen}
\end{align}
To evaluate the right-hand side, note that
\begin{align}
\sum_i w_i \tr\left[\phi_{i,0}\bcaM_\caE^{-1}(\phi_{i,0}) \right] &= \Tr\left[\sum_i w_i \dproj{\phi_{i,0}}\bcaM_\caE^{-1}\right] = \frac{1}{d}\Tr\left(\bcaM_\caE \bcaM_\caE^{-1}\right) = \frac{1}{d}\Tr(\bbfI)=\frac{d^2-1}{d}, \label{eq:TrBId}
\end{align}
where $\Tr$ denotes the trace on $\caL^{\rmH}(\caH)$ (viewing $\caM_\caE$ as a linear operator on the operator space), and $\bbfI$ denotes the identity on $\caL^{\rmH}_0(\caH)$. Combining Eqs.~\eqref{eq:LBIC}--\eqref{eq:TrBId} yields
\begin{align}
\max_{\psi\in \caP(\caH)} \|\psi_0\|_\caE^2 \ge \frac{d}{K}\left(\frac{d^2-1}{d}\right)^2 = \frac{(d^2-1)^2}{Kd},
\end{align}
which establishes the desired inequality in \pref{pro:ShNormLUB}.
\end{proof}

\subsection{Auxiliary results on the combined phase-design ensemble based on MUBs}
In this section, we clarify the basic properties of the combined phase 3-design ensemble $\caE_N$ with $2\le N\le d+1$, as introduced in the main text, and provide a proof of \thref{thm:MMSoptWCMUB}. This ensemble is a uniform mixture of $N$  phase 3-design ensembles $\caE_{\caB_j}$ associated with $N$ mutually unbiased bases (MUBs), denoted by $\caB_j=\{\ket{k}_{\caB_j}\}_{k=0}^{d-1}$ for $j=1,2,\ldots,N$. More precisely, $\caE_N$ can be expressed as follows:
\begin{align}
    \caE_N = \bigsqcup_{j=1}^{N} \frac{1}{N}\caE_{\caB_j}.
\end{align}

\begin{proposition}
\label{pro:AMoment}
The second and third moment operators of $\caE_N$ read
\begin{align}
    Q_2(\caE_N) &= \frac{2}{d^2}P_{[2]} -\frac{1}{Nd^2}\sum_{j=1}^{N}\sum_{k=0}^{d-1}(\proj{kk})_{\caB_j}
\label{eq:2ndMA}\\
    Q_3(\caE_N) &= \frac{6}{d^3}P_{[3]} -\frac{5}{Nd^3}\sum_{j=1}^{N}\sum_{k=0}^{d-1}(\proj{kkk})_{\caB_j} -\frac{3}{Nd^3}\sum_{j=1}^{N}\sum_{[\bmk]\colon |[\bmk]|=3} (\proj{[\bmk]})_{\caB_j}. \label{eq:3ndMA}
\end{align}
\end{proposition}

\begin{proposition}\label{pro:AN2design}
The second frame potential of $\caE_N$ reads
\begin{equation}
    \Phi_2(\caE_N)=\frac{2d^2-2d+1}{d^4}+\frac{d-1}{Nd^4}.
\end{equation}
The ensemble $\caE_N$ is a state $2$-design if and only if $N=d+1$.
\end{proposition}

\Pref{pro:AMoment} is a direct corollary of \lref{lem:URPAVG}, and \pref{pro:AN2design} follows 
from \pref{pro:AMoment},
 given that the bases $\caB_1, \caB_2, \ldots, \caB_N$ are mutually unbiased. When $N = d+1$, the $N$ bases form a complete set of MUBs and thus constitute a 2-design, which implies that
\begin{equation}
\frac{1}{Nd}\sum_{j=1}^{N}\sum_{k=0}^{d-1}(\proj{kk})_{\caB_j} = \frac{2P_{[2]}}{d(d+1)}.
\end{equation}
As a corollary, \pref{pro:AMoment} yields $Q_2(\caE_N) = 2P_{[2]}/[d(d+1)]$, consistent with $\caE_N$ forming a 2-design.

Next, we decompose any operator $O \in \caL^{\rmH}(\caH)$ into orthogonal parts as
\begin{align}
O =\frac{\tr(O)}{d}\bbone+O_{\perp}+ \sum_{j=1}^{N} O_{\caB_j} ,
\end{align}
where $O_{\caB_j}$ denotes the traceless diagonal part of $O$ with respect to basis $\caB_j$, and $O_{\perp}$ is the remaining orthogonal component. When $N=d+1$, we have $O_{\perp}=0$.
\begin{proposition}\label{pro:caFA}
Suppose $O\in \caL^{\rmH}(\caH)$. Then
\begin{align}
    \caM_{\caE_N}(O) &= \frac{\tr(O)}{d}\bbone+\frac{1}{d}O_{\perp} + 
    \sum_{j=1}^{N}\frac{N-1}{Nd}O_{\caB_j}, \label{eq:caMcaE_N}\\
\caM^{-1}_{\caE_N}(O) &= \frac{\tr(O)}{d}\bbone+dO_{\perp} +\sum_{j=1}^{N}\frac{Nd}{N-1}O_{\caB_j}, \label{eq:caMcaE_Ninv}
\end{align}
where $\caM_{\caE_N}$ is defined in \eqsref{eq:FD1}{eq:FD2}.
\end{proposition}

\begin{lemma}\label{lem:PI3A}
Suppose $\rho\in \caD(\caH)$, $O\in \caL^{\rmH}_0(\caH)$, and $\caB= \{\ket{k}_\caB\}_{k=0}^{d-1}$ is an orthonormal basis of $\bbC^d$. Define
\begin{align}
    \Pi_\caB(\rho, O) \coloneqq \tr\left[\sum_{[\bmk]\colon|[\bmk]|=3} |[\bmk]|(\proj{[\bmk]})_\caB\left(\rho \otimes O\otimes O\right)\right].
\end{align}
Then
\begin{align}\label{eq:PI3A}
    -\Pi_\caB(\rho, O) \le \tr\left(\rho O_\caB^2\right) + \left\|O_\caB\right\|_2^2.
\end{align}
\end{lemma}
\begin{lemma}\label{lem:NRF}
Suppose $\rho\in \caD(\caH)$ and $O\in \caL^{\rmH}_0(\caH)$. Then
\begin{align}\label{eq:7N}
    d^3\tr\left[Q_3(\caE_N)\left(\rho\otimes O\otimes O\right)\right] \le \left(3+\frac{1}{N}\right)\|O\|_2^2.
\end{align}
\end{lemma}

\begin{proof}[Proof of \pref{pro:caFA}]
Suppose the ensemble $\caE_N$ has the form $\caE_N=\{\ket{\phi_i}, w_i\}_i$. 
By \eref{eq:FD1}, for any $O\in \caL^{\rmH}(\caH)$, the measurement channel acts as
\begin{align}
    \caM_{\caE_N}(O) = d \sum_i w_i\tr(\phi_i O)\,\phi_i = d\tr_1\left[Q_2(\caE_N)\left(O\otimes \bbone\right)\right].
\end{align}
Direct calculation based on \eref{eq:2ndMA} in \pref{pro:AMoment} yields
\begin{align}
    \caM_{\caE_N}(\bbone)=\bbone,\qquad \caM_{\caE_N}(O_{\caB_j})=\frac{N-1}{Nd}\,O_{\caB_j},\qquad \caM_{\caE_N}(O_\perp)=\frac{1}{d}\,O_\perp.
\end{align}
By linearity, we have
\begin{align}
\caM_{\caE_N}(O) &= \frac{\tr(O)}{d}\bbone+\frac{1}{d}O_{\perp} + 
    \sum_{j=1}^{N}\frac{N-1}{Nd}O_{\caB_j},
\end{align}
which confirms \eref{eq:caMcaE_N}. Since $\caM_{\caE_N}(O)$ preserves the orthogonal decomposition, by inverting it on each eigenspace, we readily obtain the expression for $\caM_{\caE_N}^{-1}(O)$ shown in \eref{eq:caMcaE_Ninv}. 
\end{proof}

\begin{proof}[Proof of \lref{lem:PI3A}]
Without loss of generality, assume that $\caB$ is the computational basis and $\ket{k}_\caB=\ket{k}$ for $k=0,1,\ldots, d-1$. Then
\begin{align}
    \sum_{[\bmk]\colon|[\bmk]|=3} |[\bmk]|\proj{[\bmk]} &= \sum_{\substack{k,l \\ k \neq l}}\proj{kll} + \sum_{\substack{k,l \\ k \neq l}}\ketbra{kll}{lkl} + \sum_{\substack{k,l \\ k \neq l}}\ketbra{kll}{llk} + \sum_{\substack{k,l \\ k \neq l}}\ketbra{lkl}{kll} + \sum_{\substack{k,l \\ k \neq l}}\ketbra{lkl}{lkl} + \sum_{\substack{k,l \\ k \neq l}}\ketbra{lkl}{llk}\nonumber\\
    &\equad + \sum_{\substack{k,l \\ k \neq l}}\ketbra{llk}{kll} + \sum_{\substack{k,l \\ k \neq l}}\ketbra{llk}{lkl} + \sum_{\substack{k,l \\ k \neq l}}\ketbra{llk}{llk}.
\end{align}
Taking the trace against $\rho \otimes O \otimes O$ gives
\begin{align}
    \Pi_\caB(\rho, O) &= \sum_{\substack{k,l \\ k \neq l}}\rho_{kk}O_{ll}^2 + 2\sum_{\substack{k,l \\ k \neq l}}\rho_{ll}\left|O_{kl}\right|^2 + 4\sum_{\substack{k,l \\ k \neq l}}\Re\left(\rho_{kl}O_{lk}\right)O_{ll} + 2\sum_{\substack{k,l \\ k \neq l}}\rho_{ll}O_{ll}O_{kk}\nonumber\\
    &= 2\sum_{\substack{k,l \\ k \neq l}}\rho_{kk}O_{ll}^2 + 2\sum_{\substack{k,l \\ k \neq l}}\rho_{ll}\left|O_{lk}\right|^2 + 4\sum_{\substack{k,l \\ k \neq l}}\Re\left(\rho_{kl}O_{lk}\right)O_{ll} - \sum_{l}\rho_{ll}O_{ll}^2 - \sum_l O_{ll}^2, \label{eq:PirhoX0}
\end{align}
where the second equality uses the facts $|O_{kl}|=|O_{lk}|$, $\sum_k\rho_{kk}=\tr(\rho)=1$, and $\sum_k O_{kk}=\tr(O)=0$. The sum of the first three terms is nonnegative because 
\begin{equation}
\rho_{kk}O_{ll}^2 + \rho_{ll}\left|O_{lk}\right|^2 \geq 2\sqrt{\rho_{kk}\rho_{ll}}\,|O_{lk}O_{ll}| \geq 2|\rho_{kl}O_{lk}O_{ll}| \geq -2\Re(\rho_{kl}O_{lk})O_{ll}\quad \forall k,l,
\end{equation}
where the second inequality holds because every $2\times 2$ principal submatrix of $\rho$ with respect to $\caB$ is positive semidefinite. Therefore,
\begin{align}
    \Pi_\caB(\rho, O) \ge -\tr\left(\rho O_\caB^2\right) - \left\|O_\caB\right\|_2^2,
\end{align}
which establishes \eref{eq:PI3A}.
\end{proof}

\begin{proof}[Proof of \lref{lem:NRF}]
From \pref{pro:AMoment}, we have
\begin{align}
    Q_3(\caE_N) = \frac{6}{d^3}P_{[3]} - \frac{5}{Nd^3}\sum_{j=1}^{N}\sum_{k=0}^{d-1}(\proj{kkk})_{\caB_j} - \frac{3}{Nd^3}\sum_{j=1}^{N}\sum_{[\bmk]\colon|[\bmk]|=3} (\proj{[\bmk]})_{\caB_j}.
\end{align}
Since $O$ is traceless by assumption, we can apply the following two key bounds:
\begin{align}
    6\tr\left[P_{[3]}\left(\rho\otimes O\otimes O\right)\right] &= 2\tr(\rho O^2) + \|O\|_2^2 \le 3\|O\|_2^2, \quad 
    -\Pi_{\caB_j}(\rho, O) \le \tr\left(\rho O_{\caB_j}^2\right) + \left\|O_{\caB_j}\right\|_2^2, 
\end{align}
where the latter follows from \lref{lem:PI3A}.
Combining the above two equations, we obtain
\begin{align}
    d^3\tr\left[Q_3(\caE_N)\left(\rho\otimes O\otimes O\right)\right] &= 6\tr\left[P_{[3]}\left(\rho\otimes O\otimes O\right)\right] - \frac{5}{N}\sum_{j=1}^N\tr\left(\rho O_{\caB_j}^2\right) - \frac{1}{N}\sum_{j=1}^N\Pi_{\caB_j}(\rho,O)\nonumber\\
    &\le 3\|O\|_2^2 + \frac{1}{N}\sum_{j=1}^N\left\|O_{\caB_j}\right\|_2^2 - \frac{4}{N}\sum_{j=1}^N\tr\left(\rho O_{\caB_j}^2\right)\le \left(3+\frac{1}{N}\right)\|O\|_2^2, \label{eq:P4rhoX}
\end{align}
where the last step holds because $\sum_{j=1}^N\left\|O_{\caB_j}\right\|_2^2 \le \|O\|_2^2$ and  $\tr\bigl(\rho O_{\caB_j}^2\bigr)\geq0$ for each $j$.
\end{proof}

\subsection{Proof of \thref{thm:MMSoptWCMUB}}
\begin{proof}
By \eref{eq:ShadowNormSD} and \lref{lem:NRF}, we bound the  squared shadow norm as follows:
\begin{align}
\left\|O\right\|_{\caE_N}^2=d\max_{\rho\in \caD(\caH)}\tr\left\{Q_3(\caE) \left[\rho \otimes \caM_{\caE_N}^{-1}(O)\otimes \caM_{\caE_N}^{-1}(O)\right]\right\}\le \frac{3N+1}{Nd^2}\left\|\caM_{\caE_N}^{-1}(O)\right\|_2^2. \label{eq:N1}
\end{align}
According to \pref{pro:caFA}, $\caM_{\caE_N}^{-1}(O)$ has the form
\begin{align}
\caM_{\caE_N}^{-1}(O) = dO_{\perp}+\sum_{j=1}^N\frac{Nd}{N-1}\,O_{\caB_j} ,
\end{align}
where the identity component vanishes because $\tr(O)=0$ by assumption. These components are mutually orthogonal with respect to the Hilbert--Schmidt inner product, so
\begin{align}
\left\|\caM_{\caE_N}^{-1}(O)\right\|_2^2 = d^2\left[\left\|O_{\perp}\right\|_2^2+\frac{N^2}{(N-1)^2}\sum_{j=1}^N\left\|O_{\caB_j}\right\|_2^2\right]\le \frac{N^2 d^2}{(N-1)^2}\,\|O\|_2^2, \label{eq:N2_expand}
\end{align}
where the inequality holds because $\sum_{j=1}^N\left\|O_{\caB_j}\right\|_2^2 + \left\|O_{\perp}\right\|_2^2 = \|O\|_2^2$. 
Together with \eref{eq:N1}, this equation means
\begin{align}
\|O\|_{\caE_N}^2 \le \frac{3N+1}{Nd^2} \cdot \frac{N^2 d^2}{(N-1)^2}\,\|O\|_2^2 = \frac{N(3N+1)}{(N-1)^2}\,\|O\|_2^2,
\end{align}
which completes the proof of \thref{thm:MMSoptWCMUB}.
\end{proof}

\section{Fourth moments of Haar-random observables and Clifford orbits}
\label{sup:schur}

\subsection{Fourth moments of Haar-random observables}
In preparation for studying the average shadow norm achieved by 2-design POVMs, we recall some basic results about Schur--Weyl duality for the unitary group.

Let $\rmU(\caH)$ be the unitary group on $\caH$ and $S_t$ the symmetric group of $t$ numbers. 
According to Schur--Weyl duality, the $t$-th tensor power $\caH^{\otimes t}$ decomposes into a multiplicity-free direct sum of irreducible representations of $\rmU(\caH) \times S_t$:
\begin{align}
    \caH^{\otimes t}=\bigoplus_\lambda \caH_\lambda=\bigoplus_\lambda \caW_\lambda\otimes \caS_\lambda,
\end{align}
where each $\lambda$ labels a nonincreasing partition of $t$ into at most $d$ parts, $\caW_\lambda$ is the Weyl module carrying an irreducible representation of $\rmU(\caH)$, and $\caS_\lambda$ is the Specht module carrying an irreducible representation of $S_t$. The dimensions of $\caW_\lambda$ and $\caS_\lambda$ are denoted by $D_\lambda$ and $d_\lambda$, respectively. The projector $P_{\lambda}$ onto $\caH_\lambda$ can be expressed as
\begin{align}
    P_{\lambda}=\frac{d_\lambda}{24}\sum_{\sigma\in S_4}\chi_\lambda(\sigma)R(\sigma), \label{eq:Plambda}
\end{align}
where $\chi_\lambda(\sigma)$ is the character of $\sigma$ in the representation $\lambda$, and $R(\sigma)$ is the unitary operator representing the permutation $\sigma$. By definition, $\tr(P_\lambda)=\rk(P_\lambda)=d_\lambda D_\lambda$.

\begin{lemma}
\label{lem:4th_moment}
Suppose $\caU$ is a unitary 4-design and $O \in \caL(\caH)$. Then
\begin{align}
    \underset{U\sim\caU}{\bbE}\, U^{\otimes 4} O^{\otimes 4} U^{\dagger \otimes 4}=\sum_{\lambda}\tr\left( P_{\lambda} O^{\otimes 4}\right)\frac{P_\lambda}{d_\lambda D_\lambda},
\end{align}
where $\underset{U\sim\caU}{\bbE}$ denotes the average over $\caU$.
\end{lemma}
\begin{proof}
By assumption, $\caU$ is a unitary 4-design, which implies that
\begin{align}
\underset{U\sim\caU}{\bbE}\, U^{\otimes 4} O^{\otimes 4} U^{\dagger \otimes 4}=\underset{U\sim \mathrm{Haar}}{\bbE}\, U^{\otimes 4} O^{\otimes 4} U^{\dagger \otimes 4}.
\end{align}
Thus both sides are invariant under permutations and diagonal unitary transformations of the form $U^{\otimes 4}$ for $U\in \rmU(\caH)$. The lemma then follows directly from Schur--Weyl duality.
\end{proof}

\begin{table}[t]
\centering
\setlength{\tabcolsep}{12pt}
\caption{Dimensions of the Specht module $\caS_\lambda$ and Weyl module $\caW_\lambda$.}
\label{tab:s4_dimen}
\begin{tabular}{lc Sc}
\toprule
$\lambda$ & $d_\lambda$ & $D_\lambda$ \\
\midrule
$[4]$       & $1$ & $\dfrac{d(d+1)(d+2)(d+3)}{24}$ \\
$[1,1,1,1]$ & $1$ & $\dfrac{d(d-1)(d-2)(d-3)}{24}$ \\
$[2,2]$     & $2$ & $\dfrac{d^2(d^2-1)}{12}$ \\
$[2,1,1]$   & $3$ & $\dfrac{d(d-2)(d^2-1)}{8}$ \\
$[3,1]$     & $3$ & $\dfrac{d(d+2)(d^2-1)}{8}$ \\
\bottomrule
\end{tabular}
\end{table}

\begin{table}[t]
\centering
\setlength{\tabcolsep}{10pt}
\renewcommand{\arraystretch}{1.2}
\caption{Characters of the symmetric group $S_4$.}
\label{tab:s4_cha}
\begin{tabular}{lccccc}
\toprule
Cycle type & $(1^4)$ & $(2^2)$ & $(2,1^2)$ & $(3,1)$ & $(4)$ \\
\midrule
Order              & $1$  & $2$  & $2$  & $3$  & $4$ \\
No.\ of elements   & $1$  & $3$  & $6$  & $8$  & $6$ \\
\midrule
$\chi_1 = [4]$        & $1$  & $1$  & $1$  & $1$  & $1$ \\
$\chi_2 = [1,1,1,1]$  & $1$  & $1$  & $-1$ & $1$  & $-1$ \\
$\chi_3 = [2,2]$      & $2$  & $2$  & $0$  & $-1$ & $0$ \\
$\chi_4 = [2,1,1]$    & $3$  & $-1$ & $-1$ & $0$  & $1$ \\
$\chi_5 = [3,1]$      & $3$  & $-1$ & $1$  & $0$  & $-1$ \\
\bottomrule
\end{tabular}
\end{table}

In the rest of this section, we focus on the case $t = 4$.  There are five potential partitions: $[4]$, $[2,2]$, $[2,1,1]$, $[3,1]$, and $[1,1,1,1]$, where $[2,1,1]$ is relevant when $d\geq 3$, and $[1,1,1,1]$ is relevant when $d\geq 4$. The dimensions of the Specht and Weyl modules for these partitions are presented in \tref{tab:s4_dimen}, and the corresponding characters in \tref{tab:s4_cha}. Using \tref{tab:s4_cha} and \eref{eq:Plambda}, we obtain the following lemma.

\begin{lemma}
\label{lem:moment_tr}
Suppose $O \in \caL^{\rmH}_0(\caH)$ and $\phi_i, \phi_j \in \caP(\caH)$. Then
    \begin{alignat}{2}
        \tr\left( P_{[4]} O^{\otimes 4}\right)&=\frac{\|O\|_2^4+2\|O\|_4^4}{8},&\quad  \tr\left[P_{[4]} \left(\phi_i^{\otimes 2}\otimes\phi_j^{\otimes 2}\right)\right]&=\frac{\left[\tr(\phi_i\phi_j)\right]^2+4\tr(\phi_i\phi_j)+1}{6}, \nonumber\\
        \tr\left( P_{[2,2]} O^{\otimes 4}\right)&=\frac{\|O\|_2^4}{2},&\quad   \tr\left[ P_{[2,2]} \left(\phi_i^{\otimes 2}\otimes\phi_j^{\otimes 2}\right)\right]&=\frac{\left[\tr(\phi_i\phi_j)\right]^2-2\tr(\phi_i\phi_j)+1}{3}, \nonumber\\
        \tr\left( P_{[1,1,1,1]} O^{\otimes 4}\right)&=\frac{\|O\|_2^4-2\|O\|_4^4}{8},&\quad\tr\left[ P_{[1,1,1,1]} \left(\phi_i^{\otimes 2}\otimes\phi_j^{\otimes 2}\right)\right]&=0,\\
        \tr\left( P_{[2,1,1]} O^{\otimes 4}\right)&=\frac{-3\|O\|_2^4+6\|O\|_4^4}{8},&\quad \tr\left[ P_{[2,1,1]}\left(\phi_i^{\otimes 2}\otimes\phi_j^{\otimes 2}\right)\right]&=0, \nonumber\\
        \tr\left( P_{[3,1]} O^{\otimes 4}\right)&=\frac{-3\|O\|_2^4-6\|O\|_4^4}{8},&\quad\tr\left[ P_{[3,1]}\left(\phi_i^{\otimes 2}\otimes\phi_j^{\otimes 2}\right)\right]&=\frac{-\left[\tr(\phi_i\phi_j)\right]^2+1}{2}.\nonumber
    \end{alignat}
\end{lemma}

\subsection{Fourth moments of Clifford orbits}
In this appendix, we focus on an $n$-qubit system, whose Hilbert space $\caH$ has dimension $d = 2^n$.

Let $\{I,X,Y,Z\}^{\otimes n}$ be the set of $n$-qubit Pauli operators. Then $\{W^{\otimes 4}\mid W\in \{I,X,Y,Z\}^{\otimes n}\}$ forms a stabilizer group with stabilizer projector
\begin{equation}
  P_n=\frac{1}{d^2}\sum_{W\in \{I,X,Y,Z\}^{\otimes n}} W^{\otimes 4}.
\end{equation}
For a pure state $\psi \in \caP(\caH)$, the stabilizer 2-R\'enyi entropy \cite{Leone22} is defined as
\begin{align}
   M_2(\psi)\coloneqq-\log_2\sum_{W\in\{I,X,Y,Z\}^{\otimes n}}\left[\tr(W\psi)\right]^4+\log_2 d = -\log_2\tr\left(P_n\psi^{\otimes 4}\right)-\log_2 d.
\end{align}
By definition, $\tr\left(P_n\psi^{\otimes 4}\right)=2^{-M_2(\psi)}/d$. As established in \rcite{Zhu16}, the stabilizer 2-R\'enyi entropy satisfies
\begin{align}
      0\le M_2(\psi)\le \log_2(d+1)-1 \quad \forall \psi\in \caP(\caH), \label{eq:stab_bound}
\end{align}
where the lower bound is saturated if and only if $\psi$ is a stabilizer state. Under the action of the Clifford group, the symmetric subspace of $\caH^{\otimes 4}$ decomposes into two inequivalent irreducible components with projectors
\begin{align}
P_+=P_{[4]}P_n=P_nP_{[4]},\quad P_-=P_{[4]}(\bbone-P_n)=(\bbone-P_n)P_{[4]},
\end{align}
and dimensions
\begin{align}
  D^{+}_{[4]}=\tr(P_+)=\frac{(d+1)(d+2)}{6},\quad  D^{-}_{[4]}=\tr(P_-)=\frac{(d^2-1)(d+2)(d+4)}{24}. \label{eq:D4pm}
\end{align}
By construction, we have
\begin{align}
\tr\left(P_+\psi^{\otimes 4}\right)=\tr\left(P_n\psi^{\otimes 4}\right)=\frac{2^{-M_2(\psi)}}{d},\quad \tr\left(P_-\psi^{\otimes 4}\right)=1-\frac{2^{-M_2(\psi)}}{d}\quad \forall \psi\in \caP(\caH). \label{eq:PpPmProb}
\end{align}

\begin{lemma}
\label{lem:4th_tr_Cliff}
Suppose $\psi, \phi \in \caP(\caH)$. Then
    \begin{align}
\underset{U\sim \Cl(n)}{\bbE}\left(U \psi U^{\dagger}\right)^{\otimes 4}
    &=\frac{2^{-M_{2}(\psi)}}{d\,D^{+}_{[4]}}\,P_++\frac{1}{D^{-}_{[4]}}\left(1-\frac{2^{-M_{2}(\psi)}}{d}\right)P_-,  \label{eq:4thTwirlCliff}\\
\underset{U\sim \Cl(n)}{\bbE}\left[\tr\left(\phi\, U \psi U^{\dagger}\right)\right]^4
    &=\frac{2^{-M_{2}(\psi)-M_{2}(\phi)}}{d^2\,D^{+}_{[4]}}+\frac{1}{D^{-}_{[4]}}\left(1-\frac{2^{-M_{2}(\phi)}}{d}\right)\left(1-\frac{2^{-M_{2}(\psi)}}{d}\right)\leq \frac{5(d+3)}{4(d+4)D_{[4]}}\leq \frac{5}{4D_{[4]}}, \label{eq:4thTwirlClifftr}
    \end{align}
where $D^{+}_{[4]}$ and $D^{-}_{[4]}$ are given in \eref{eq:D4pm}, and $D_{[4]}$ is given in \tref{tab:s4_dimen}.
\end{lemma}
\begin{proof}[Proof of \lref{lem:4th_tr_Cliff}]
\Eref{eq:4thTwirlCliff} follows from Schur's lemma and \eref{eq:PpPmProb}; it was essentially proved in \rcite{Zhu16}, though without the concept of stabilizer 2-R\'enyi entropy. The equality in \eref{eq:4thTwirlClifftr} is a direct corollary of \eqsref{eq:PpPmProb}{eq:4thTwirlCliff}. The first inequality in \eref{eq:4thTwirlClifftr} follows from \eref{eq:stab_bound}, and the second is trivial.
\end{proof}

The following lemma is a simple corollary of \lref{lem:4th_tr_Cliff}.
\begin{lemma}
\label{lem:4th_poten_Cliff}
Suppose $\caE=\{\ket{\phi_{i}}, w_i\}_i$ is a state ensemble on $\caH$, and $\psi \in \caP(\caH)$ is a pure state. Then
     \begin{align}
         \underset{U\sim \Cl(n)}{\bbE}\Phi_4\left(\caE,U\psi U^{\dagger}\right)&
    =\frac{\sum_i 2^{-M_{2}(\phi_i)-M_{2}(\psi)}w_i}{d^2\,D^{+}_{[4]}}+\frac{1}{D^{-}_{[4]}}\left(1-\frac{\sum_i 2^{-M_{2}(\phi_i)}w_i}{d}\right)\left(1-\frac{2^{-M_{2}(\psi)}}{d}\right)\leq \frac{5(d+3)}{4(d+4)D_{[4]}}\leq \frac{5}{4D_{[4]}}.
    \end{align}
\end{lemma}

\section{Proofs of results on average-case optimal shadow estimation}
\label{sup:avg_norm_up}

In this section, we prove our  main results on optimal shadow estimation in the average-case setting, including  \thref{thm:AvgShNormUB}  and \pref{pro:ShNormFidLD} in the main text  as well as \pref{pro:AShNormFidUB} in the End Matter. 

\subsection{Auxiliary results on shadow norms}
\begin{lemma}
\label{lem:MaxEigUB}
Suppose $A$ is a positive semidefinite operator acting on $\caH$ with $\|A\|_1 = a$ and $\|A\|_2^2 = b$, where $a$ and $b$ are positive numbers satisfying ${a^2}/{d}\le b \le a^2$. Then
\begin{align}
\|A\|\leq\frac{a + \sqrt{(d-1)(db - a^2)}}{d}, \label{eq:MaxEigUB}
\end{align}
and the inequality is saturated when all eigenvalues of $A$, except for the largest one, are equal.
\end{lemma}
\begin{proof}[Proof of \lref{lem:MaxEigUB}]
Let $\mu_1 \ge \mu_2 \ge \dots \ge \mu_d \ge 0$ be the eigenvalues of $A$ in nonincreasing order; then $\|A\|= \mu_1$. By definition,
\begin{align}
\sum_{i=1}^d \mu_i =\|A\|_1= a,\quad  \quad \sum_{i=1}^d \mu_i^2 =\|A\|_2^2= b.
\end{align}
Applying the Cauchy--Schwarz inequality, we deduce that
\begin{align}
(a - \mu_1)^2=\left( \sum_{i=2}^d \mu_i \right)^2 \le (d-1) \sum_{i=2}^d \mu_i^2=(d-1)\left(b - \mu_1^2\right),
\end{align}
which simplifies to
\begin{align}
d\mu_1^2 - 2a\mu_1 + \left[a^2 - (d-1)b\right] &\le 0.
\end{align}
This quadratic inequality implies that
\begin{align}
\|A\|=\mu_1\le \frac{a + \sqrt{(d-1)(db - a^2)}}{d},
\end{align}
which confirms \eref{eq:MaxEigUB}. If $\mu_2 = \mu_3 = \dots = \mu_d$, then all three inequalities above are saturated, which completes the proof of \lref{lem:MaxEigUB}.
\end{proof}

\begin{lemma}\label{lem:4NormO2Norm}
Suppose $O\in \caL^{\rmH}_0(\caH)$. Then
\begin{equation}
\frac{1}{d}\|O\|_2^4 \leq \|O\|_4^4\leq \frac{d-1}{d}\|O\|_2^4.  \label{eq:4NormO2Norm}
\end{equation}
\end{lemma}
\begin{proof}
The first inequality in \eref{eq:4NormO2Norm} follows from the Cauchy--Schwarz inequality. To prove the second inequality, denote by $O_+$ and $O_-$ the positive and negative parts of $O$, respectively. Then, by assumption, we have
\begin{equation}
  \tr(O_+)=\tr(O_-),\quad  \frac{1}{d-1} \|O_+\|_2^2\leq \|O_-\|_2^2\leq (d-1)\|O_+\|_2^2.
\end{equation}
Without loss of generality, we can assume that $\|O_-\|_2^2\leq \|O_+\|_2^2=1$; then $\|O\|_2^2=\|O_+\|_2^2+\|O_-\|_2^2\geq d/(d-1)$.  In addition, the eigenvalues of $O$ are bounded in absolute value by 1, so that
\begin{equation}
   \|O\|_4^4=\|O_+\|_4^4+\|O_-\|_4^4\leq \|O_+\|_2^2+\|O_-\|_2^2= \|O\|_2^2 \leq \frac{d-1}{d}\|O\|_2^4,
\end{equation}
which confirms the second inequality in \eref{eq:4NormO2Norm} and completes the proof of \lref{lem:4NormO2Norm}.
\end{proof}

\begin{proposition}
\label{pro:AvgShNormTUB}
   Suppose $\caE$ is a state $2$-design, $O \in \caL^{\rmH}_0(\caH)$, and $\caU$ is a unitary $4$-design. Then
   \begin{align}    
   \left\|\Xi(O,\caU)\right\|_\caE^2 \le \left(\frac{d+1}{d}+\sqrt{f(d,r)\bPhi_3(\caE)+g(d,r)}\lsp\right)\|O\|_2^2,    \label{eq:AvgShNormTUB}
   \end{align}
   where
   \begin{equation}
r=\frac{\|O\|_4^4}{\|O\|_2^4}, \quad  f(d,r)=\frac{12(d+1)^2\left[(d^2+3d+3)+(d^2+d)r\right]}{d^2(d+2)^2(d+3)},\quad g(d,r)=\frac{4(d+1)^2\left[(d^2-3d)r-(4d+3)\right]}{d^2(d+2)(d+3)}.     \label{eq:func}
   \end{equation}
\end{proposition}
By definition in \eref{eq:func} and \lref{lem:4NormO2Norm}, it is straightforward to verify that
\begin{equation}\label{eq:func2}
\frac{1}{d}\le r\le \frac{d-1}{d},\quad
f(d,r) \le \frac{24}{d} - \frac{26}{d^2},\quad
g(d,r) \le 4-\frac{g_d}{d},\quad
g_d\coloneqq\begin{cases}
18 & d=2,\\
22 & d=3,\\
25 & d=4,5,\\
29 &d\geq 6,\\
30 &d\geq 7.
\end{cases}
\end{equation}

\begin{proof}[Proof of \pref{pro:AvgShNormTUB}]
Suppose the ensemble $\caE$ has the form $\caE=\{\ket{\phi_i},w_i\}_i$ and let $A\coloneqq\sum_i w_i\left[\tr\left(\phi_i O\right)\right]^2 \phi_i$. Then
\begin{align}
    \|A\|_1=\frac{1}{d(d+1)}\|O\|_2^2,\quad \|A\|_2^2=\sum_{i,j}w_iw_j\left[\tr(O\phi_{i}) \right]^2\left[\tr(O\phi_{j})\right]^2\tr(\phi_{i}\phi_{j}).
\end{align}
Using \lref{lem:MaxEigUB}, we deduce that
\begin{align}
     &\|O\|_\caE^2\le (d+1)^2\sqrt{d(d-1) \sum_{i,j}w_iw_j\left[\tr(O\phi_{i}) \right]^2\left[\tr(O\phi_{j})\right]^2\tr(\phi_{i}\phi_{j})-\frac{d-1}{d^2(d+1)^2}\|O\|_2^4}+\frac{d+1}{d}\|O\|_2^2.
\end{align}
Applying Jensen's inequality, we further derive an upper bound for $\underset{U\sim\caU}{\bbE}\left\|UOU^{\dagger}\right\|_\caE^2$:
\begin{align}
     \underset{U\sim\caU}{\bbE}\left\|UOU^{\dagger}\right\|_\caE^2&\le (d+1)^2\sqrt{d(d-1)  \underset{U\sim\caU}{\bbE}\sum_{i,j}w_iw_j\left[\tr(UOU^{\dagger}\phi_{i}) \right]^2\left[\tr(UOU^{\dagger}\phi_{j})\right]^2\tr(\phi_{i}\phi_{j})-\frac{d-1}{d^2(d+1)^2}\|O\|_2^4} \nonumber\\
     &\equad+\frac{d+1}{d}\|O\|_2^2. \label{eq:AvgShNormTUBProof1}
\end{align}
The expectation under the square root can be evaluated as follows:
\begin{align}
    &\underset{U\sim\caU}{\bbE} \sum_{i,j}w_iw_j\left[\tr(UOU^{\dagger}\phi_{i}) \right]^2\left[\tr(UOU^{\dagger}\phi_{j})\right]^2\tr(\phi_{i}\phi_{j})= \sum_{i,j}w_iw_j\tr\left[\underset{U\sim\caU}{\bbE}U^{\otimes 4}O^{\otimes 4}U^{\dagger \otimes 4}\left( \phi_i^{\otimes 2}\otimes\phi_j^{\otimes 2}\right)\right]\tr(\phi_{i}\phi_{j})\nonumber\\
    &=\sum_{i,j}\sum_{\lambda}\frac{w_iw_j}{d_\lambda D_\lambda}\tr\left( P_{\lambda} O^{\otimes 4}\right)\tr\left(P_\lambda\phi_i^{\otimes 2}\otimes\phi_j^{\otimes 2}\right)\tr(\phi_{i}\phi_{j})=h_1(d,O)\Phi_1+h_2(d,O)\Phi_2+h_3(d,O)\Phi_3, \label{eq:AvgNormTUBProof2}
\end{align}
where $\lambda$ represents a nonincreasing partition of 4 into at most $d$ parts (see SM \sref{sup:schur}) and
\begin{equation}
    \begin{aligned}
h_1(d,O)&=\frac{(d^2+3d+6)\left[\tr(O^2)\right]^2-4d\tr(O^4)}{d^2(d-1)(d+1)(d+2)(d+3)},\\
h_2(d,O)&=\frac{-(12d+12)\left[\tr(O^2)\right]^2+(4d^2-4d)\tr(O^4)}{d^2(d-1)(d+1)(d+2)(d+3)},\\
h_3(d,O)&=\frac{(2d^2+6d+6)\left[\tr(O^2)\right]^2+(2d^2+2d)\tr(O^4)}{d^2(d-1)(d+1)(d+2)(d+3)}.
\end{aligned}
\end{equation}
The second equality in \eref{eq:AvgNormTUBProof2} follows from \lref{lem:4th_moment}, and the third equality is derived using \lref{lem:moment_tr}. 

Combining \eqsref{eq:AvgShNormTUBProof1}{eq:AvgNormTUBProof2}, we deduce that
\begin{align}
    \left\|\Xi(O,\caU)\right\|_\caE^2 &\le\frac{d+1}{d}\|O\|_2^2 +\sqrt{d(d+1)^4(d-1)\left[h_1(d,O)\Phi_1+h_2(d,O)\Phi_2+h_3(d,O)\Phi_3\right]-\frac{(d+1)^2(d-1)}{d^2}\|O\|_2^4}\nonumber\\
    &=\left(\frac{d+1}{d}+\sqrt{f(d,r)\bPhi_3(\caE)+g(d,r)}\lsp\right)\|O\|_2^2, \label{eq:avg_G_up}
\end{align}
where $\Phi_j$ is shorthand for the frame potential $\Phi_j(\caE)$ for $j=1, 2, 3$, and the equality holds because $\Phi_1=1/d$ and $\Phi_2=2/[d(d+1)]$ (since $\caE$ is a state 2-design) together with the definition $r=\|O\|_4^4/\|O\|_2^4$. This completes the proof of \pref{pro:AvgShNormTUB}.
\end{proof}

\subsection{Proof of \thref{thm:AvgShNormUB}}
\begin{proof}[Proof of \thref{thm:AvgShNormUB}]
By virtue of \pref{pro:AvgShNormTUB} and the inequalities in \eref{eq:func2}, we deduce that
\begin{align}  \left\|\Xi(O,\caU)\right\|_\caE^2 &\le \left(\frac{d+1}{d} + \sqrt{\left(\frac{24}{d}-\frac{26}{d^2}\right)\bPhi_3(\caE) + 4 - \frac{g_d}{d}}\lsp\right)\|O\|_2^2\le \left(\frac{d+1}{d} + \sqrt{\frac{24}{d}\bPhi_3(\caE) + 4 - \frac{26}{d^2} -\frac{g_d}{d}}\lsp\right)\|O\|_2^2\nonumber\\
    &\le \left(\frac{d+1}{d} + \sqrt{\frac{24}{d}\bPhi_3(\caE) + 4 - \frac{30}{d}}\lsp\right)\|O\|_2^2\leq \left(1 + \sqrt{\frac{24}{d}[\bPhi_3(\caE)-1] + 4 }\lsp\right)\|O\|_2^2\nonumber\\
    &\le \left(2\sqrt{2}+1\right)\|O\|_2^2.
\end{align}
Here the second inequality holds because $\bPhi_3(\caE)\ge 1$, the third holds because $(26/d)+g_d\geq 30$ by \eref{eq:func2}, the last holds because $\bPhi_3(\caE) \le (d + 2)(d^2 + 2d - 1)/(6d^2)$ by \pref{pro:FP3UB}, which implies that $24[\bPhi_3(\caE)-1]/d\leq 4$, and the fourth inequality follows from the concavity of the square root:
\begin{align}
\sqrt{\frac{24}{d}\bPhi_3(\caE) + 4 - \frac{30}{d}}\leq
\sqrt{\frac{24}{d}[\bPhi_3(\caE)-1] + 4 }-\frac{3}{d\sqrt{\frac{24}{d}[\bPhi_3(\caE)-1] + 4 }} \leq
\sqrt{\frac{24}{d}[\bPhi_3(\caE)-1] + 4 }-\frac{1}{d}.
\end{align}
This completes the proof of \thref{thm:AvgShNormUB}.
\end{proof}

\subsection{Auxiliary results on shadow norms for fidelity estimation}\label{app:FidEaux}
Suppose $\caE=\{\ket{\phi_i},w_i\}_i$ is a state 2-design on $\caH$, $\psi \in \caP(\caH)$, and $\psi_0=\psi-{\bbone}/{d}$. To better understand the properties of the squared shadow norm $\|\psi_0\|_\caE^2$ tied to fidelity estimation, we introduce some auxiliary results. Define
\begin{equation} J(\caE,\psi)\coloneqq d(d + 1)^2\sum_iw_i\left[\tr\left(\phi_i\psi\right)\right]^2\phi_{i},\quad \eta(\caE,\psi)\coloneqq\left\|J(\caE,\psi)\right\|.  \label{eq:JN}
\end{equation}

\begin{proposition}
\label{pro:w}
     Suppose $\caE$ forms a state 2-design on $\caH$ and $\psi\in\caP(\caH)$. Then
\begin{align}
   \eta(\caE,\psi)-3-\frac{2}{d}\le\left\|\psi_0\right\|_\caE^2=\left\|J(\caE,\psi)-\frac{2(d+1)}{d}\psi\right\|-1+\frac{1}{d^2}\le \eta(\caE,\psi)-1+\frac{1}{d^2}.\label{eq:w}
\end{align}
\end{proposition}

\begin{proposition}
\label{pro:concentration_fidelity}
    Suppose $\caE$ forms a state 2-design on $\caH$ and $\psi \in \caP(\caH)$. Then
    \begin{align}
   \dfrac{6(d+1)}{d+2}\bPhi_3(\caE,\psi)\leq      \eta(\caE,\psi)&\le \lsp \dfrac{d+1}{\sqrt{(d+2)(d+3)}}\sqrt{48\bPhi_4(\caE,\psi)}\leq \sqrt{48\bPhi_4(\caE,\psi)},\label{eq:QnormLUB}\\
         \dfrac{6(d+1)}{d+2}\bPhi_3(\caE,\psi)-3-\frac{2}{d}\leq \left\|\psi_0\right\|_\caE^2&\le \dfrac{d+1}{\sqrt{(d+2)(d+3)}}\sqrt{48\bPhi_4(\caE,\psi)}-1+\frac{1}{d^2}\leq \sqrt{48\bPhi_4(\caE,\psi)}-1.\label{eq:EpsiShNormLUB}
    \end{align}
\end{proposition}

\begin{proof}[Proof of \pref{pro:w}]
By assumption, $\caE$ is a state 2-design, so $\caM_\caE^{-1}(O)=(d+1)O$ whenever $O$ is traceless. By virtue of \eqref{eq:norm2} with $O=\psi_0=\psi-\bbone/d$, we deduce that
\begin{align}
   \left\|\psi_0\right\|_\caE^2&=\left\|d(d+1)^2\sum_iw_i\left[\tr\left (  \phi_{i}\psi_0\right )\right]^2\phi_{i}\right\| \nonumber\\&=\left\|d(d+1)^2\sum_iw_i\left[\tr\left (  \phi_i\psi\right )\right]^2\phi_i-2(d+1)^2\sum_iw_i\tr\left (  \phi_i\psi\right )\phi_i+\dfrac{(d+1)^2}{d}\sum_iw_i\phi_i\right\|
  \nonumber\\
  &=\left\|J(\caE,\psi)-\frac{2(d+1)}{d}\psi-\left(1-\frac{1}{d^2}\right)\bbone\right\|=\left\|J(\caE,\psi)-\frac{2(d+1)}{d}\psi\right\|-1+\frac{1}{d^2},
\label{eq:CRFE_norm_fidelity}
\end{align}
which confirms the equality in \eref{eq:w}. Here the third equality holds because $\sum_iw_i\phi_i={\bbone}/{d}$ and $\sum_iw_i\tr\left(\phi_i\psi\right)\phi_i={(\bbone+\psi)}/{[d(d+1)]}$, given that $\caE$ is a state 2-design. In addition, the operator $J(\caE,\psi)-2(d+1)\psi/d$ is positive semidefinite, which implies that 
\begin{equation}
\eta(\caE,\psi)-  \frac{2(d+1)}{d} \leq  \left\|J(\caE,\psi)-\frac{2(d+1)}{d}\psi\right\|\leq \eta(\caE,\psi).
\end{equation}
The above two equations together confirm \eref{eq:w} and complete the proof of \pref{pro:w}.
\end{proof}

\begin{proof}[Proof of \pref{pro:concentration_fidelity}]
The first inequality in \eref{eq:QnormLUB} can be proved as follows:
\begin{align}
    \eta(\caE,\psi)&=\max_{\rho}{d(d+1)^2}\sum_iw_i\tr(\phi_i\rho)\left[\tr(\phi_i \psi)\right]^2\ge {d(d+1)^2}\sum_iw_i\left[\tr(\phi_i \psi)\right]^3\nonumber\\
    &=d(d+1)^2\Phi_3(\caE,\psi)=\dfrac{6(d+1)}{d+2}\bPhi_3(\caE,\psi), \label{eq:w_lower}
\end{align}
where the three equalities hold by definition. By the Cauchy--Schwarz inequality, we further deduce that
\begin{align}
         \eta(\caE,\psi)&=\max_{\rho}{d(d+1)^2}\sum_iw_i\tr(\phi_i\rho)\left[\tr(\phi_i \psi)\right]^2 \le  \max_{\rho}d(d+1)^2\sqrt{\Phi_2(\caE,\rho)}\sqrt{\Phi_4(\caE,\psi)}\nonumber\\
    &=\sqrt{2}\lsp d^{1/2}(d+1)^{3/2}\sqrt{\Phi_4(\caE,\psi)}=\dfrac{(d+1)}{\sqrt{(d+2)(d+3)}}\sqrt{48\bPhi_4(\caE,\psi)}\le\sqrt{48\bPhi_4(\caE,\psi)},     \label{eq:w_upper}
    \end{align}
where the second equality holds because $\max_\rho \Phi_2(\caE,\rho)= 2/[d(d+1)]$, given that $\caE$ forms a state 2-design by assumption, and the third equality holds by definition. The above two equations together confirm \eref{eq:QnormLUB}.

Next, the first two inequalities in \eref{eq:EpsiShNormLUB} follow from \pref{pro:w} and \eref{eq:QnormLUB}. The last inequality in \eref{eq:EpsiShNormLUB} is a simple corollary of the following inequalities:
\begin{align}
   &\left(1- \dfrac{d+1}{\sqrt{(d+2)(d+3)}}\right)\sqrt{48\bPhi_4(\caE,\psi)}\ge \dfrac{1}{d+2}\sqrt{48\bPhi_4(\caE,\psi)}\ge\dfrac{4}{d+2}> \dfrac{1}{d^2},
\end{align}
where the second inequality holds because $\bPhi_4(\caE,\psi)\geq 1/3$ by \lref{lem:FPrelation}. This completes the proof of \pref{pro:concentration_fidelity}.
\end{proof}

\subsection{Proofs of \pref{pro:ShNormFidLD} and \pref{pro:AShNormFidUB}}
\begin{proof}[Proof of \pref{pro:ShNormFidLD}]
\Eref{eq:ShNormFidUB} is a simple corollary of \thref{thm:AvgShNormUB} given that $\|\psi_0\|_2^2=(d-1)/d$. To prove \eref{eq:ShNormFidLD}, note that $\left\|\psi_0\right\|_\caE^2\le\sqrt{48\bPhi_4(\caE,\psi)}-1$ for all $\psi\in \caP(\caH)$ by \pref{pro:concentration_fidelity}.
If $ \left\|\psi_0 \right\|_\caE^2\ge \sqrt{48[1+k\xi(\caE)]}-1$ for some $k>0$, then
\begin{align}
\bPhi_4(\caE,\psi)-\bbE \bPhi_4(\caE,\psi) =  \bPhi_4(\caE,\psi)-1\geq k \xi(\caE)>k \sqrt{(d+6)^3\Var(\bPhi_4(\caE,\psi))},
\end{align}
where the last inequality follows from \lref{lem:4th_poten}. Therefore, by Chebyshev's inequality, we have
\begin{align}
    \Pr \left\{ \left\|\psi_0 \right\|_\caE^2 \ge \sqrt{48[1+k\xi(\caE)]}-1\right\} \le \Pr \left\{\bPhi_4(\caE,\psi)-1\ge k \xi(\caE)\right\}
\le\frac{1}{(d+6)^3k^2}\le \frac{1}{d^3k^2}\quad \forall k>0,
\end{align}
which confirms \eref{eq:ShNormFidLD}. Finally, the inequality in \eref{eq:Xi} holds because $\bPhi_3(\caE) \le (d + 2)(d^2 + 2d - 1)/(6d^2)$ by \pref{pro:FP3UB}, which completes the proof of \pref{pro:ShNormFidLD}.
\end{proof}

\begin{proof}[Proof of \pref{pro:AShNormFidUB}]
To prove \eref{eq:AShNormRP} in \pref{pro:AShNormFidUB}, without loss of generality, we can assume that $\psi\sim\caT$, where $\caT$ is the URP ensemble with respect to the computational basis.
By \pref{pro:concentration_fidelity} and \lref{lem:RP_URP}, we can deduce that
\begin{align}
     \underset{\psi\sim\caT}{\bbE}\left\|\psi_0\right\|_\caE^2\le \underset{\psi\sim\caT}{\bbE}\sqrt{48\bPhi_4(\caE,\psi)}-1\le \sqrt{\underset{\psi\sim\caT}{\bbE}48\bPhi_4(\caE,\psi)}-1\le24\sqrt{\frac{2D_{[4]}}{d^4}}=\sqrt{\frac{48(d+1)(d+2)(d+3)}{d^3}},
\end{align}
which confirms \eref{eq:AShNormRP}.

Next, we turn to  \eref{eq:AShNormFidUBCli}, assuming that $d=2^n$. By \pref{pro:concentration_fidelity} and \lref{lem:4th_poten_Cliff}, we can deduce that
\begin{align}
\underset{U\sim \Cl(n)}{ \bbE}\left\|U\psi_0U^\dagger\right\|_\caE^2&\le \sqrt{2}\lsp d^{1/2}(d+1)^{3/2}\underset{U\sim \Cl(n)}{\bbE}\sqrt{\Phi_4(\caE,U\psi U^{\dagger})}-1+\frac{1}{d^2}\nonumber\\&\le \sqrt{2}\lsp d^{1/2}(d+1)^{3/2}\sqrt{\underset{U\sim \Cl(n)}{\bbE}\Phi_4(\caE,U\psi U^{\dagger})}-1+\frac{1}{d^2} \nonumber\\
&\leq \sqrt{2}\lsp d^{1/2}(d+1)^{3/2}\sqrt{\frac{5(d+3)}{4(d+4)D_{[4]}}}-1+\frac{1}{d^2}\nonumber\\
&=\sqrt{\frac{60(d+1)^2}{(d+2)(d+4)}}-1+\frac{1}{d^2}\le \frac{2\sqrt{15}\lsp(d+1)}{d+2}-1+\frac{1}{d^2}\le 2\sqrt{15}-1,
    \end{align}
which confirms \eref{eq:AShNormFidUBCli} and completes the proof of \pref{pro:ShNormFidLD}.
\end{proof}

\section{Further numerical results on mean squared shadow norms}
\label{sup:numerics}

In this appendix, we present additional numerical results on the mean squared shadow norms for several families of observables and state ensembles.

\subsection{Shadow norms for fidelity estimation of Haar-random states}
\label{subsec:haar_numerics}

\begin{figure}[b]
    \centering
    \includegraphics[width=0.55\textwidth]{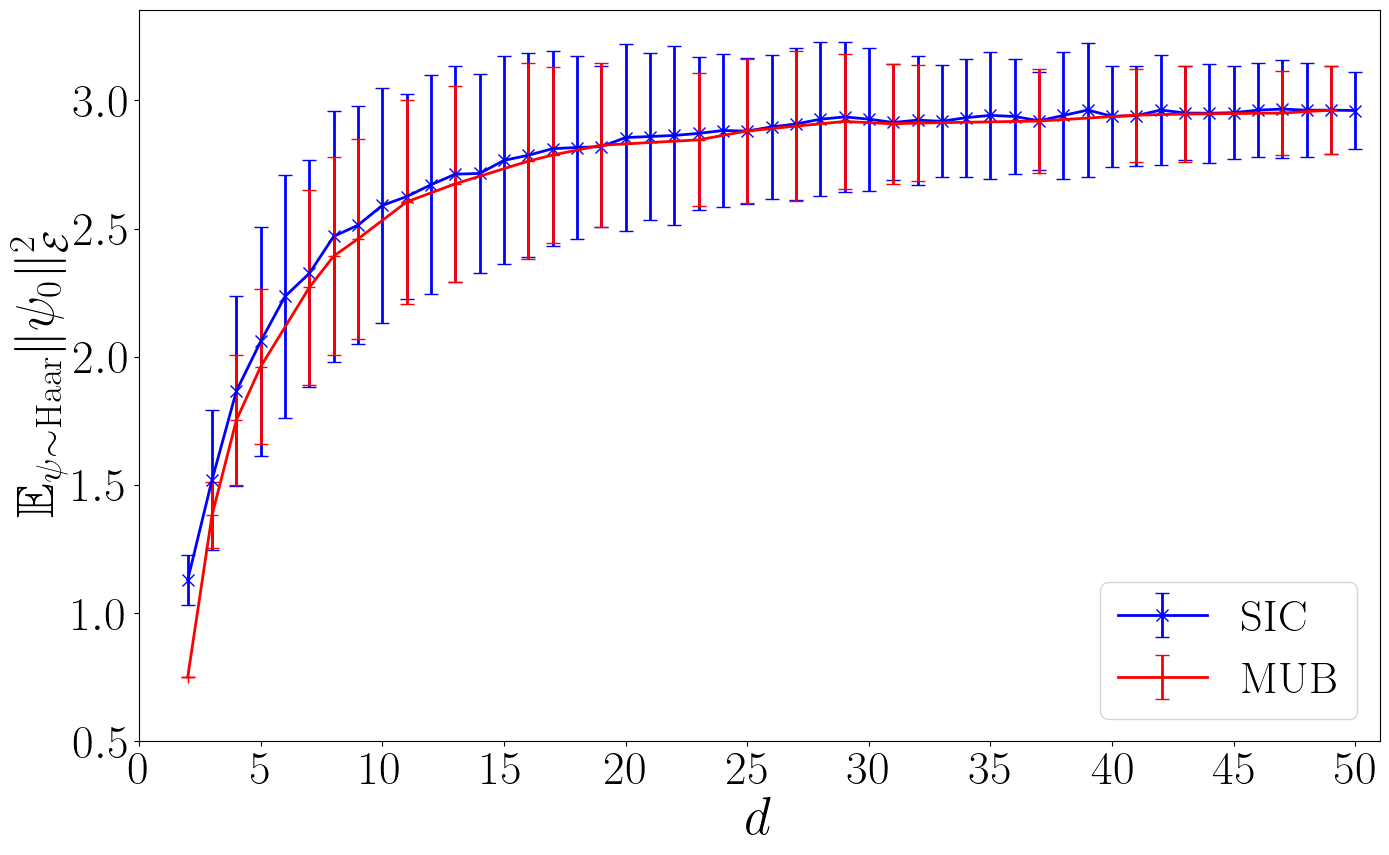}
    \vspace{-8pt}
    \caption{Mean squared shadow norms $\bbE_{\psi\sim\Ha}\|\psi_0\|_{\caE}^2$ for fidelity estimation of Haar-random pure states with SIC-POVMs (blue) and MUBs (red). Error bars denote the standard deviation over 20\,000 random pure states.}
    \label{fig:concentration1}
\end{figure}

Here we provide additional results on shadow norms for fidelity estimation based on SIC-POVMs and complete sets of MUBs (available for prime-power dimensions). As in the main text, SIC-POVMs are generated by the Heisenberg--Weyl group from fiducial states listed in Ref.~\cite{Scott10}; complete sets of MUBs are constructed according to Refs.~\cite{Wootters89,DURT10}. To complement \fref{fig:AShNormFid} in the main text, \fref{fig:concentration1} provides a refined comparison between SIC-POVMs and MUBs for Haar-random pure states, with error bars quantifying statistical fluctuations. Both measurement schemes exhibit convergence of the squared shadow norms toward their ensemble averages with increasing $d$, indicating suppressed variability in higher dimensions. MUBs yield slightly narrower error bars at small $d$, while the dispersion becomes comparable to that of SIC-POVMs at large $d$, implying asymptotically similar performance.

\subsection{Shadow norms for fidelity estimation of states in Clifford orbits}
\label{subsec:clifford_numerics}

Next, we consider fidelity estimation of states in orbits of the $n$-qubit Clifford group using SIC-POVMs. \Fref{fig:CliffSIC} displays the ensemble averages $\bbE_{U\sim\Cl(n)}\bigl\|U\psi_0 U^{\dagger}\bigr\|_{\caE}^2$ for the stabilizer orbit and a random Clifford orbit. The mean values for the two orbits nearly coincide, especially for $n\geq 3$. SIC-POVMs are generated by the Heisenberg--Weyl group from fiducial states listed in \rcite{Scott10}; for the special case $n=7$ ($d=128$), in which no fiducial is tabulated in that reference, we adopt the online database linked in Ref.~\cite{Fuchs17}.

\begin{figure}[b]
    \centering
    \includegraphics[width=0.55\textwidth]{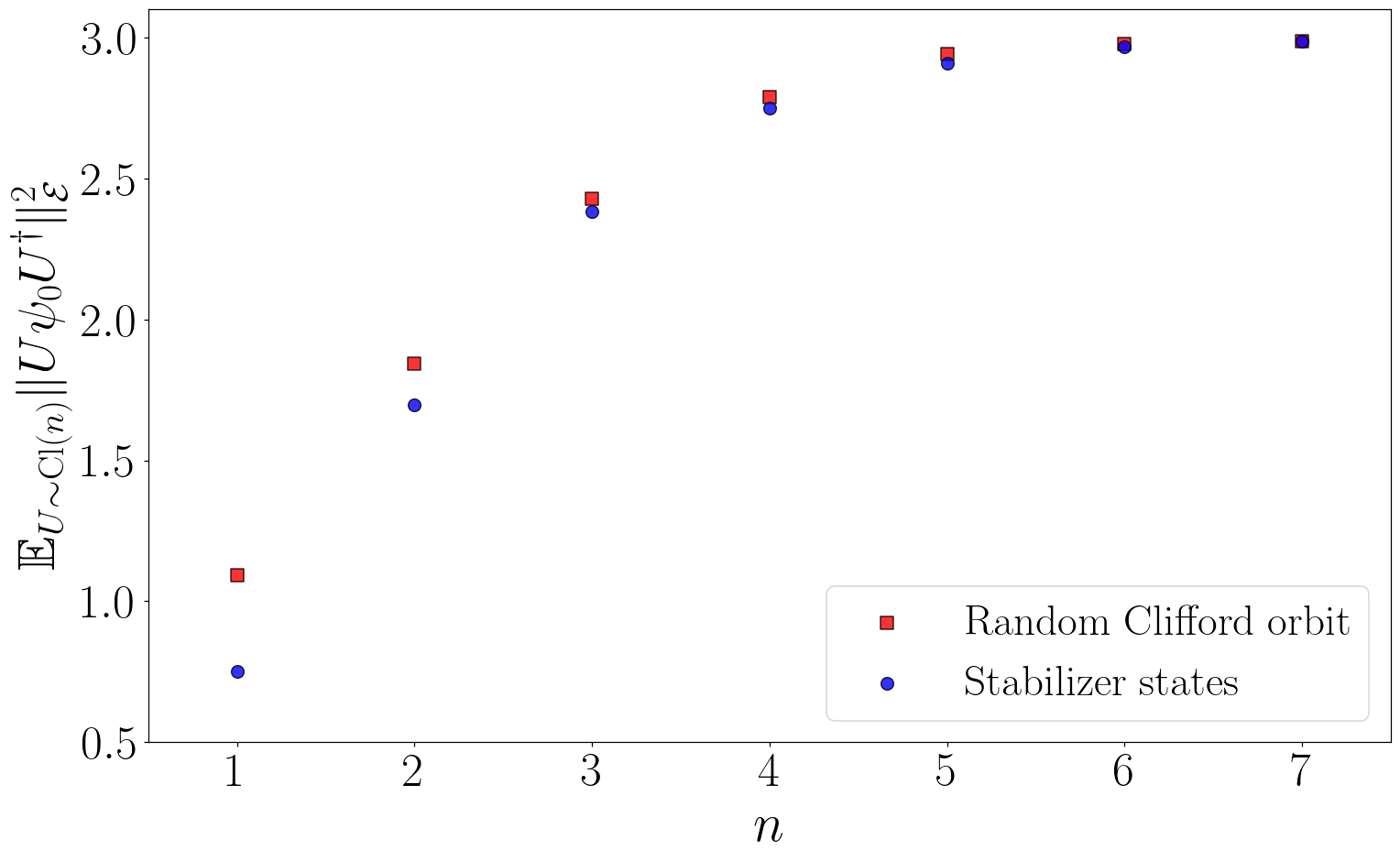}
    \vspace{-6pt}
    \caption{Mean squared shadow norm $\bbE_{U\sim\Cl(n)}\bigl\|U\psi_0 U^{\dagger}\bigr\|_{\caE}^2$ for fidelity estimation with SIC-POVMs. Results on the stabilizer orbit (blue circles) and a random Clifford orbit (red squares) nearly coincide, especially for $n\geq 5$. For each $n$, 2000 random Clifford unitaries are sampled.}
    \label{fig:CliffSIC}
\vspace{4ex}

\includegraphics[width=0.55\textwidth]{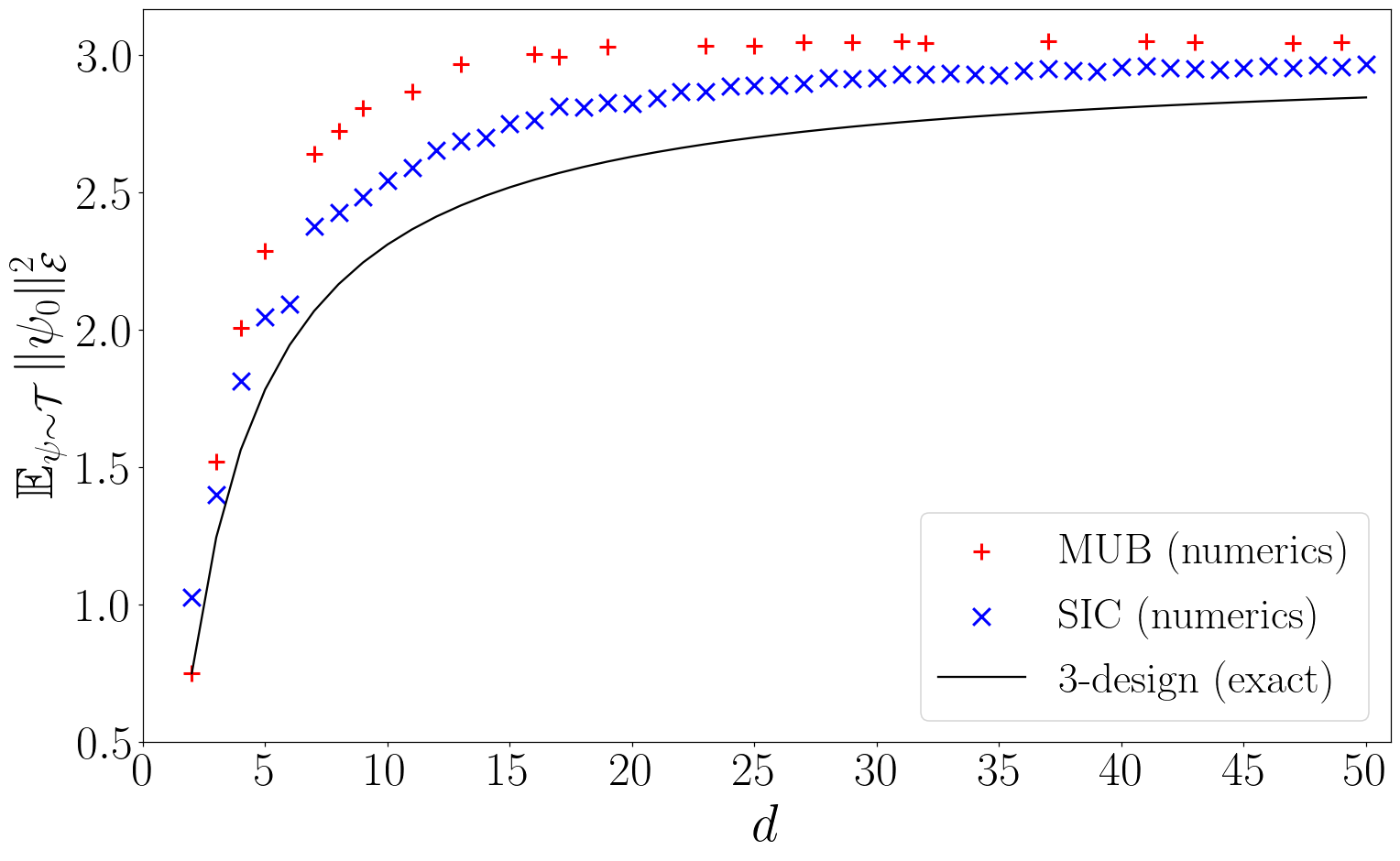}
    \vspace{-6pt}
    \caption{Mean squared shadow norm $\bbE_{\psi\sim\caT}\bigl\|\psi_0\bigr\|_{\caE}^2$ over the URP ensemble for fidelity estimation with SIC-POVMs and MUBs. The result achieved by an exact 3-design is also shown as a benchmark.}
    \label{fig:avg_URP}
\end{figure}

\subsection{Shadow norms for fidelity estimation of uniform random phase states}
\label{subsec:urp_numerics}

Next, we consider the ensemble $\caT$ of uniform random phase (URP) states defined over the computational basis. \Fref{fig:avg_URP} displays the mean squared shadow norms $\bbE_{\psi\sim\caT}\bigl\|\psi_0\bigr\|_{\caE}^2$ for both SIC-POVMs and MUBs (one of the bases coincides with the computational basis). Both schemes exhibit similar scaling behavior, with SIC-POVMs maintaining a slight advantage across all dimensions shown except for $d=2$.

\subsection{Shadow norms of observables from the Gaussian unitary ensemble}
\label{subsec:gue_numerics}

\begin{figure}[b]
    \centering
    \includegraphics[width=0.55\textwidth]{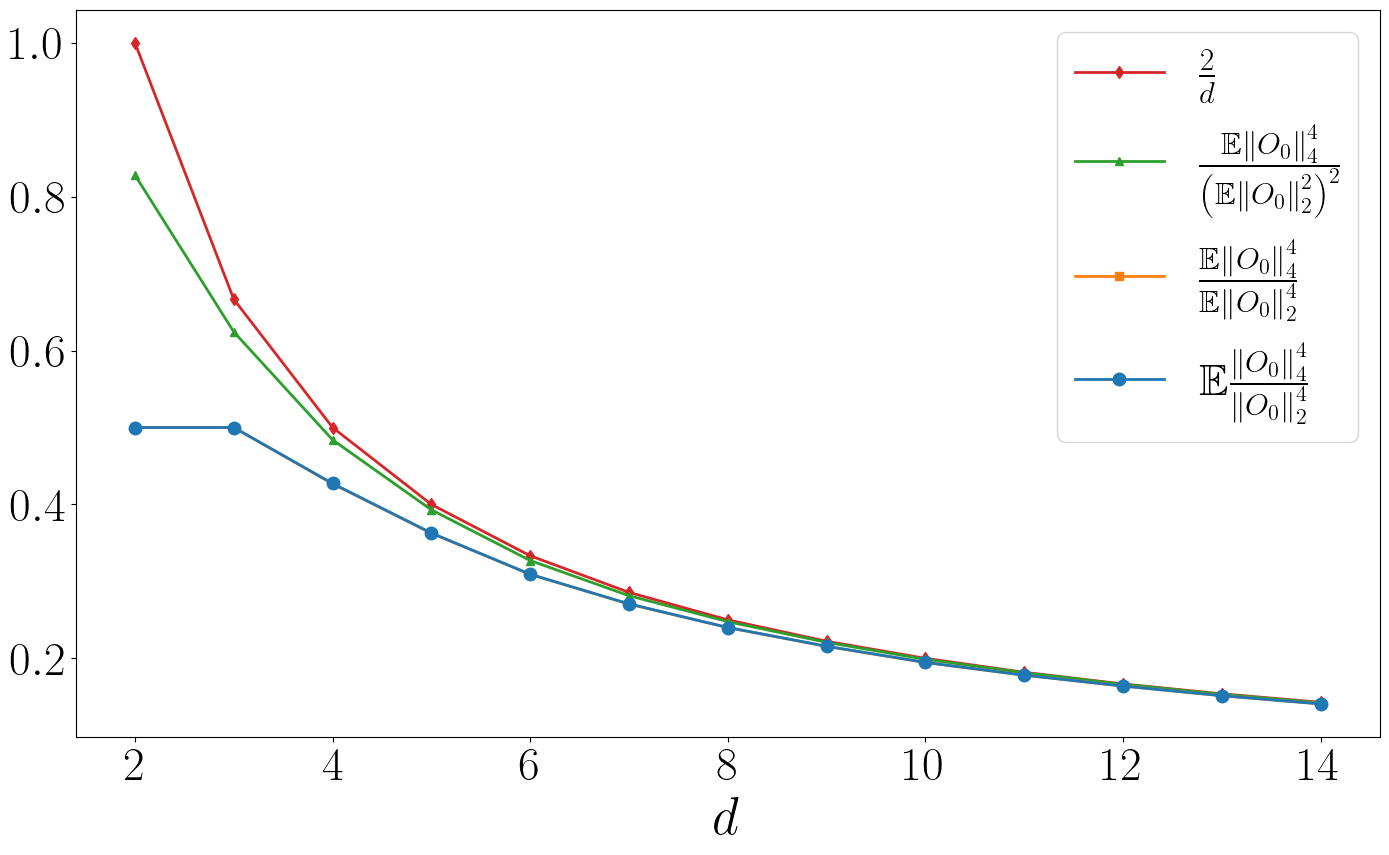}
    \vspace{-4pt}
    \caption{Comparison of $\bbE\bigl[\|O_0\|_4^4/\|O_0\|_2^4\bigr]$, $\bbE\lsp\|O_0\|_4^4\big/\bigl(\bbE\lsp\|O_0\|_2^2\bigr)^2$, and $\bbE\lsp\|O_0\|_4^4\big/\bbE\lsp\|O_0\|_2^4$ for GUE observables. For each dimension, 2000 observables are sampled. The curve $2/d$ is shown as a benchmark.}
    \label{fig:m_0}
    \vspace{4ex}
    \includegraphics[width=0.55\textwidth]{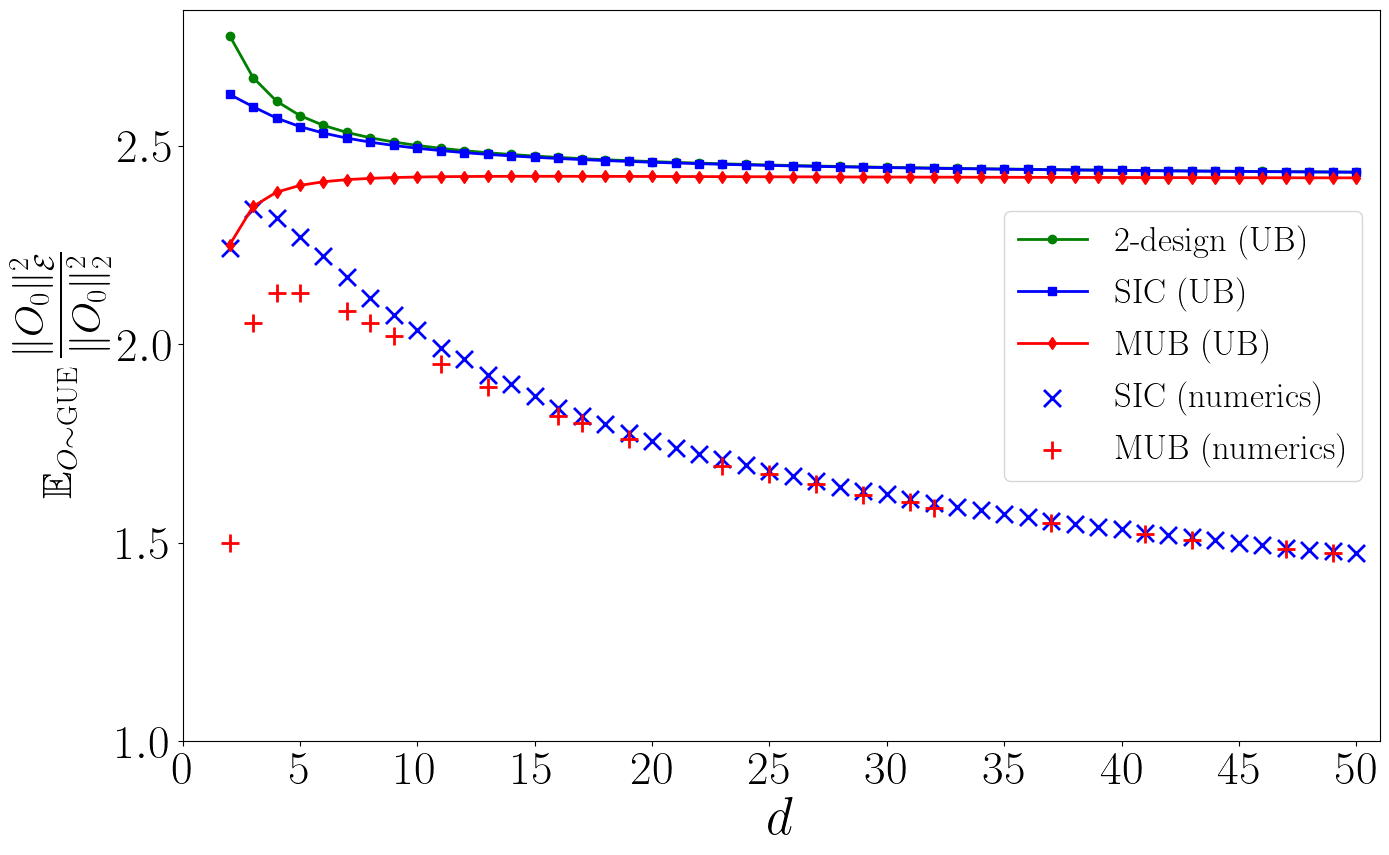}
    \vspace{-7pt}
    \caption{Mean squared shadow norms over normalized traceless GUE observables versus dimension $d$. Scatter points show numerical results over 2000 GUE samples for SIC-POVMs and MUBs (the latter restricted to prime-power dimensions). Lines with markers show analytical upper bounds from \pref{pro:avg_norm_Gauss} with $r=2/d$  for the worst-case 2-design (\pref{pro:FP3UB}), SIC-POVMs, and MUBs.}
    \label{fig:avg_fi_Gauss}
\end{figure}

Finally, we extend the analysis beyond pure-state observables to Hermitian operators sampled from the Gaussian unitary ensemble (GUE). For a state 2-design $\caE$, we introduce the normalized mean squared shadow norm $\bbE_{O\sim\mathrm{GUE}}\bigl[\|O_0\|_{\caE}^2/\|O_0\|_2^2\bigr]$, where  $O_0 = O - \tr(O)\,\bbone/d$ and the expectation is taken over the GUE. The following result parallels \pref{pro:AvgShNormTUB}.

\begin{proposition}\label{pro:avg_norm_Gauss}
Suppose $\caE$ is a state $2$-design on $\caH$. Then
\begin{align}\label{eq:avg_norm_Gauss_bound}
    \bbE_{O\sim\mathrm{GUE}}\frac{\|O_0\|_{\caE}^2}{\|O_0\|_2^2} \le \frac{d+1}{d} + \sqrt{f(d,r)\,\bPhi_3(\caE) + g(d,r)}\,,
\end{align}
where the expectation is well defined since $\|O_0\|_2 > 0$ almost surely under the GUE,
\begin{align}\label{eq:r_def}
    r \coloneqq \bbE_{O\sim\mathrm{GUE}}\frac{\|O_0\|_4^{4}}{\|O_0\|_2^{4}}\,,
\end{align}
and $f(d,r)$, $g(d,r)$ are defined in \eref{eq:func} of \pref{pro:AvgShNormTUB}.
\end{proposition}

Numerical results in \fref{fig:m_0} suggest that the mean moment ratio $r$ converges to $2/d$ as $d$ increases, with negligible deviation for $d\geq 10$.  \Fref{fig:avg_fi_Gauss} presents mean squared shadow norms over normalized traceless GUE observables for SIC-POVMs and MUBs. As benchmarks, the figure also displays the analytical upper bounds from \pref{pro:avg_norm_Gauss} evaluated at $r=2/d$ for the worse 2-design (\pref{pro:FP3UB}) as well as SIC-POVMs and MUBs.

\end{document}